\documentclass[a4paper,fleqn]{cas-sc}

\usepackage{color}
\newcommand{\blue}{\textcolor{black}}
\newcommand{\cyan}{\textcolor{black}}
\newcommand{\red}{\textcolor{black}}

\usepackage[normalem]{ulem}

\usepackage[numbers]{natbib}

\def\tsc#1{\csdef{#1}{\textsc{\lowercase{#1}}\xspace}}
\tsc{WGM}
\tsc{QE}

\usepackage{amsfonts, amsmath, amssymb}  
\usepackage{amsthm}
\usepackage{geometry}
\usepackage{pdfpages}
\usepackage{hyperref}
\hypersetup{
    colorlinks, 
    citecolor=RubineRed,
}
\usepackage{mathtools}
\usepackage{float}
\usepackage[center]{caption}
\usepackage{subcaption}
\usepackage{listings}
\usepackage{booktabs}
\usepackage{multirow}
\usepackage{tabularx}
\usepackage{hhline}
\usepackage{graphicx}

\definecolor{codegreen}{rgb}{0,0.6,0}
\definecolor{codegray}{rgb}{0.5,0.5,0.5}
\definecolor{codepurple}{rgb}{0.58,0,0.82}
\definecolor{backcolour}{rgb}{0.95,0.95,0.92}

\lstdefinestyle{mystyle}{
    backgroundcolor=\color{backcolour},   
    commentstyle=\color{codegreen},
    keywordstyle=\color{blue},
    numberstyle=\tiny\color{codegray},
    stringstyle=\color{codepurple},
    basicstyle=\ttfamily\footnotesize,
    breakatwhitespace=false,         
    breaklines=true,                 
    captionpos=b,                    
    keepspaces=true,                 
    numbers=left,                    
    numbersep=5pt,                  
    showspaces=false,                
    showstringspaces=false,
    showtabs=false,                  
    tabsize=2
}

\newtheorem{theorem}{Theorem}
\newtheorem{lemma}[theorem]{Lemma}
\newtheorem{proposition}[theorem]{Proposition}
\newtheorem{corollary}[theorem]{Corollary}

\newcommand{\D}{d}
\newcommand{\N}{\mathcal{N}}
\newcommand{\T}{\mathcal{T}}
\newcommand{\x}{\textbf{\textup{x}}}

\newcommand{\child}[2]{c_#1(#2)}
\newcommand{\parent}[2]{p_#1(#2)}
\newcommand{\childIndex}[2]{ \text{\ifthenelse{#1=1}{\(j_1\)}{\(j_2\)}}}
\newcommand{\parentIndex}[2]{\text{\ifthenelse{#1=1}{\(j^1\)}{\(j^2\)}}}

\allowdisplaybreaks
\allowbreak

\begin{document}
\let\WriteBookmarks\relax
\def\floatpagepagefraction{1}
\def\textpagefraction{.001}

\shorttitle{A QUBO Formulation for the Tree Containment Problem}    
\shortauthors{M.J. Dinneen, P.S. Ghodla, S. Linz}  

\title[mode = title]{A QUBO Formulation for the Tree Containment Problem}

\author[1]{Michael J. Dinneen}%
[orcid=0000-0001-9977-525X]
\ead{mjd@cs.auckland.ac.nz} 

\author[1]{Pankaj S. Ghodla}[]
\ead{pgho580@aucklanduni.ac.nz} 

\author[1]{Simone Linz}%
[orcid=0000-0003-0862-9594]
\ead{s.linz@auckland.ac.nz}
\cormark[1]

\address[1]{School of Computer Science, University of Auckland, Auckland, New Zealand}

\cortext[cor1]{Corresponding author} 

\begin{abstract}[S U M M A R Y]
Phylogenetic (evolutionary) trees and networks are leaf-labeled graphs that are widely used to represent the evolutionary relationships between entities such as species, languages, cancer cells, \blue{and} viruses. To reconstruct and analyze phylogenetic networks, the problem of deciding whether or not a given rooted phylogenetic network embeds a given rooted phylogenetic tree is of recurring interest. This problem, formally know as Tree Containment, is NP-complete in general and polynomial-time solvable for certain classes of phylogenetic networks. In this paper, we connect ideas from quantum computing and phylogenetics to present an efficient Quadratic Unconstrained Binary Optimization formulation for Tree Containment in the general setting.
For an instance \((\N,\T)\) of Tree Containment, where $\N$ is a phylogenetic network with  $n_\N$ vertices and $\T$ is a phylogenetic tree with $n_\T$ vertices, the number of logical qubits  that are required for our formulation is \(O(n_\N n_\T)\).
\end{abstract}

\begin{keywords}
Quantum computing\sep
QUBOs\sep
Phylogenetic trees and networks\sep
Tree Containment
\end{keywords}

\maketitle


\section{Introduction} \label{sec: introduction}

Phylogenetics is the study of evolutionary histories and relationships between different, often biological, entities such as species, genes, viruses, or languages that are generically referred to as taxa. Traditionally, rooted leaf-labeled trees, which are known as phylogenetic trees, have been widely used to represent and analyze evolutionary relationships that are dominated by tree-like processes like speciation \cite{felsenstein2004inferring}. A phylogenetic tree is reconstructed from biological sequence data (e.g. DNA or protein \blue{sequences}) under various optimization criteria or evolutionary models so that the resulting tree $\T$ `in some sense' best explains a given dataset. Each leaf of $\T$ is labeled by a taxon whereas all inner vertices of $\T$ are unlabeled. The latter can be thought of as hypothetical ancestors, including extinct species, for which no data is available. Although phylogenetic trees are an extensively used in \blue{evolutionary} biology and researchers are now able to reconstruct such trees for thousands of taxa \cite{jetz2012global}, phylogenetic trees cannot represent complex evolutionary relationships that are the result of non-treelike processes such as hybridization \blue{or} horizontal gene transfer, which are common in many groups of organisms \cite{ottenburghs2019multispecies,richardson2007horizontal,soucy2015horizontal}. For example, it has been observed that horizontal gene transfer contributes to about \(10\%\)--\(20\%\) of all genes in prokaryotes \cite{koonin2001horizontal}. Non-treelike processes, collectively referred to as reticulation, result in species whose DNA is a mosaic of DNA derived from different ancestors. To accurately describe complex evolutionary histories that include reticulation, rooted leaf-labeled graphs, called phylogenetic networks \cite{bapteste2013networks,blais2021past,huson2010phylogenetic}, are now widely acknowledged to complement phylogenetic trees.

A particularly well-studied problem, which arises in the analysis of rooted
phylogenetic networks through the lens of rooted phylogenetic trees, is the
following embedding problem.  Given a rooted phylogenetic tree $\T$
and a rooted phylogenetic network
$\N$ that have both been reconstructed for the same set of taxa, does
$\N$ embed $\T$?  This decision problem, which we will make more precise
in the next section, is called  Tree Containment. Without imposing
any structural constraints on $\N$, Tree Containment is NP-complete
\cite{kanj2008seeing}. However, it has also been shown that Tree
Containment is polynomial-time solvable, for example, for so-called
tree-child or  level $k$-networks \cite{van2010locating}. Afterwards,
Gunawan et al.~\cite{gunawan2015locating}, and Bordewich and Semple
\cite{bordewich2016reticulation} independently showed that Tree Containment
is solvable in cubic time for reticulation-visible networks, a superclass of
tree-child networks. Since then, various algorithms have been developed to
solve Tree Containment, with the fastest such algorithms having a running
time that is linear in the number of taxa. For example, see a recent
linear-time algorithm to solve Tree Containment for reticulation-visible
networks~\cite{weller2018linear} and references therein. Although these
substantial improvements have turned Tree Containment into one of the
most well-studied problems in \blue{mathematical} research on phylogenetic networks, much less
is known about how to `efficiently' solve Tree Containment for arbitrary
rooted phylogenetic networks. \red{In this case, the current best algorithms are fixed-parameter tractable algorithms, where, for a given phylogenetic network $\N$, the parameter is either the number of vertices in $\N$ whose in-degree is at least two or the treewidth of $\N$~\cite{van2022embedding,van2018unrooted}.}

\blue{In this paper, we take a quantum computing approach to Tree Containment.} Quantum computers are known to be able to solve certain problems in
significantly lower time complexity than corresponding current-best
classical algorithms. Although it is still debatable whether or not
quantum computers can solve NP-hard problems in polynomial time, they
have provided efficient solutions for several instances of NP-hard 
problems~\cite{calude2017solving,calude2017qubo,LucasNP,mahasinghe2019solving,mcgeoch2014adiabatic}.

In a talk about simulating the quantum mechanical process,
Feynman argued about simulating physics using a quantum computer
\cite{feynman2018simulating}. This talk sparked interest in building a
quantum computer. Soon after Feyman's talk, Deutsch \cite{deutsch1985quantum}
developed the universal quantum model of computation called the
\textit{Quantum Gate Model} \cite{mcgeoch2014adiabatic}. The central goal
behind this new model of computation was to exploit the properties of quantum
mechanics to get a \textit{quantum-speedup} over the classical model of
computation. Research conducted so far on this topic shows much promise with
the two primary outstanding examples being Grover's and Shor's algorithms
that we briefly describe next. Grover's algorithm \cite{grover1996fast}
searches through an unsorted database of size \(N\), of which only one record
satisfies a particular property, in \(O(\sqrt{N})\) steps. In contrast,
any classical algorithm  to solve this database problem certainly takes
\(O(N)\) steps as it needs to iterate through a significant fraction (on
average $ N/2$) of all records. Shor's algorithm \cite{shor1994algorithms}
factorizes integers. Its computational complexity is polynomial in the
number of digits of the integer to be factorized. On the other hand, no
classical algorithm is known that factorizes integers in polynomial time.

\textit{Adiabatic Quantum Computing} (AQC) is an alternative to the quantum gate model that was first described in \cite{farhi2000quantum}. Subsequently, Aharonov et al.~\cite{aharonov2008adiabatic} have developed an adiabatic simulation for any given quantum algorithm, which implies that AQC and the quantum gate model are polynomially equivalent. AQC is based on the evolution of a time-dependent Hamiltonian that transitions from an initial Hamiltonian to a final Hamiltonian. The initial Hamiltonian's ground state is easy to construct, and the final Hamiltonian's ground state encodes the solution to a given problem. For a detailed explanation of how this evolution between initial and final Hamiltonian occurs, we refer the interested reader to \cite{farhi2000quantum}. The primary advantage of AQC over the quantum gate model is that a relatively easy to build \textit{quantum annealer} \cite{mcgeoch2014adiabatic}, \blue{which} is based on AQC, can be used to identify the minimum of an objective function.

D-Wave Systems Inc. is a Canadian quantum computing company that has developed the following quantum annealers for AQC. 

\begin{center} \label{temp}
	\begin{tabular}{ |c|c|c| } 
	\hline
		\blue{Annealer} & Number of qubits & Number of couplers \\
	\hline
		D-Wave One (2011) & 128 & 352 \\
		D-Wave Two (2012) & 512 & 1,472 \\
		D-Wave 2X (2015) & 1000 & 3,360 \\
		D-Wave 2000Q (2017) & 2048 & 6,016 \\
		D-Wave Advantage (2019) & 5640 & 40,484 \\
	\hline
	\end{tabular}
\end{center}
In the table above, the number of qubits refers to the number of physical qubits available and the number of couplers refers to the number of connections between the physical qubits in a quantum annealer. Currently, both D-Wave 2000Q and D-Wave Advantage can be accessed through the D-Wave's website\footnote{Website: \url{https://www.dwavesys.com/}} and can solve problems for which an equivalent \textit{Ising} or \textit{Quadratic Unconstrained Binary Optimization (QUBO)} formulation exists. 

The main contribution of this paper is a QUBO formulation of Tree Containment for when the input is not restricted to a particular class of rooted phylogenetic networks. To the best of our knowledge, this is the first time that unconventional computing is used to approach a problem from phylogenetics. \blue{For two recent surveys on potential future applications of quantum computing  in computational biology, we refer the reader to~\cite{fedorov2021towards,outeiral2021prospects}.}

The remainder of the paper is organized as follows. Section \ref{sec: definitions} contains the basic definitions from phylogenetics, followed by a brief introduction to QUBO and the general methodology in Section \ref{sec: QUBO}. Then, in Section \ref{sec: TC problem}, we present \blue{a} reduction of Tree Containment to QUBO. We finish with an example and some experimental results in Section \ref{sec: results}, and a short conclusion in Section \ref{sec: conclusion}.

\section{Preliminaries from Phylogenetics} \label{sec: definitions}
This section introduces notation and terminology from phylogenetics. Throughout the paper, $X$ denotes a non-empty finite set. Furthermore, all logarithms are base 2, and we write \( \lg(x)\) to refer to \(\log_2(x)\).

Let \(G\) be a directed graph, with vertex set \(V(G)\) and edge set \(E(G)\). Let $u$ and $v$ be two vertices of $G$. We say that $u$ is a {\it parent} of $v$ \blue{and that $v$ is a {\it child} of $u$} if $(u,v)\in E(G)$. 
A \textit{directed path} of $G$ is a sequence $(v_1,v_2,\ldots,v_k)$ of distinct vertices in $V(G)$ such that $(v_i,v_{i+1})\in E(G)$ for each $i\in\{1,2,\ldots,k-1\}$. 
Similarly, a {\it path} of $G$ is a sequence $(v_1,v_2,\ldots,v_k)$ of distinct vertices in $V(G)$ such that $(v_i,v_{i+1})$ or $(v_{i+1},v_i)$ is an element in $E(G)$ for each $i\in\{1,2,\ldots,k-1\}$.
We say that \(G\) is \textit{weakly connected} (or short, {\it
connected}) if there is a path between any two vertices in 
\blue{$G$}. 
A vertex $u \in V(G)$ is called a \textit{root} if $u$ has in-degree~0 and there exists a directed path from \(u\) to \(v\) 
for all \(v \in V(G)\backslash\{u\}\). 
Furthermore, $x \in V(G)$ is called a \textit{terminal vertex} if it has 
out-degree~0. 
Lastly, with \((v_1,v_2,v_3)\) being a directed path in $G$ such that $v_2$ has in-degree~1 and out-degree~1, 
the operation of \textit{suppressing} \(v_2\) in $G$ results in a new directed graph with vertex set \(V(G) \backslash \{v_2\} \) and edge set \((E(G) \backslash \{(v_1,v_2),(v_2,v_3)\}) \cup \{(v_1,v_3)\} \).

We now turn to a particular \blue{class of} directed graphs that will play an important role in what follows. A \textit{rooted binary phylogenetic network} \(\N\) on \(X\) is a rooted acyclic directed graph with no two edges in parallel that satisfies the following three properties.
\begin{enumerate}
	\item The (unique) root has in-degree~0 and out-degree~2.
	\item A vertex with out-degree~0 has in-degree~1, and the set of vertices with out-degree~0 is \(X\).
	\item All remaining vertices have either in-degree~1 and out-degree~2, or in-degree~2 and out-degree~1.
\end{enumerate}
We call $X$ the {\it leaf set} of $\N$. For technical reasons, if $|X|
= 1$, then we  allow $\N$ to consist of the single vertex in $X$.
Let $\N$ be a rooted binary phylogenetic network. A vertex of $\N$
is called a \textit{tree vertex} if it has out-degree~2. Similarly, a vertex of $\N$ is called
a \textit{reticulation vertex} if it has in-degree~2 and out-degree~1. To
illustrate, see Figure~\ref{fig: phylogenetic network and trees}(a) for an
example of a rooted binary phylogenetic network with one reticulation vertex,
four tree vertices (one is also root), and four leaves. We call \(\N\) a \textit{rooted
binary phylogenetic \(X\)-tree} if \(\N\) is a rooted binary phylogenetic
network with no reticulation vertex. To ease reading, we will refer to a
rooted binary phylogenetic network and a rooted binary phylogenetic tree as
a {\it phylogenetic network} and a {\it phylogenetic tree}, respectively,
since all such networks and trees in this paper are rooted and binary.

Let $\N_1$ and $\N_2$ be two phylogenetic networks on $X$ with vertex and edge sets $V_1$ and $E_1$, and $V_2$ and $E_2$, respectively. We say that $\N_1$ is {\em isomorphic} to $\N_2$ if there is a bijection $\varphi: V_1\rightarrow V_2$ such that $\varphi(x)=x$ for all $x\in X$, and $(u, v)\in E_1$ if and only if $(\varphi(u), \varphi(v))\in E_2$ for all $u, v\in V_1$.

\begin{figure}[h]
    \centering
    \begin{subfigure}[b]{0.3\textwidth}
        \centering
        \includegraphics[width=\textwidth]{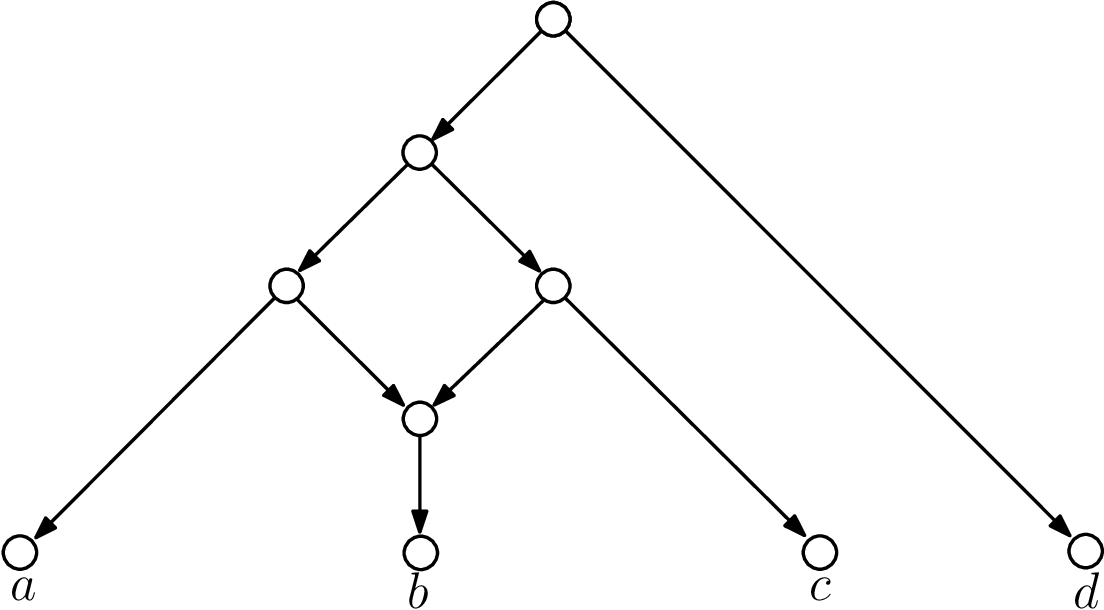}
        \caption{\(\N\)}
        \label{fig: example of phylogenetic network}
    \end{subfigure}
    \hfill
    \begin{subfigure}[b]{0.3\textwidth}
        \centering
        \includegraphics[width=\textwidth]{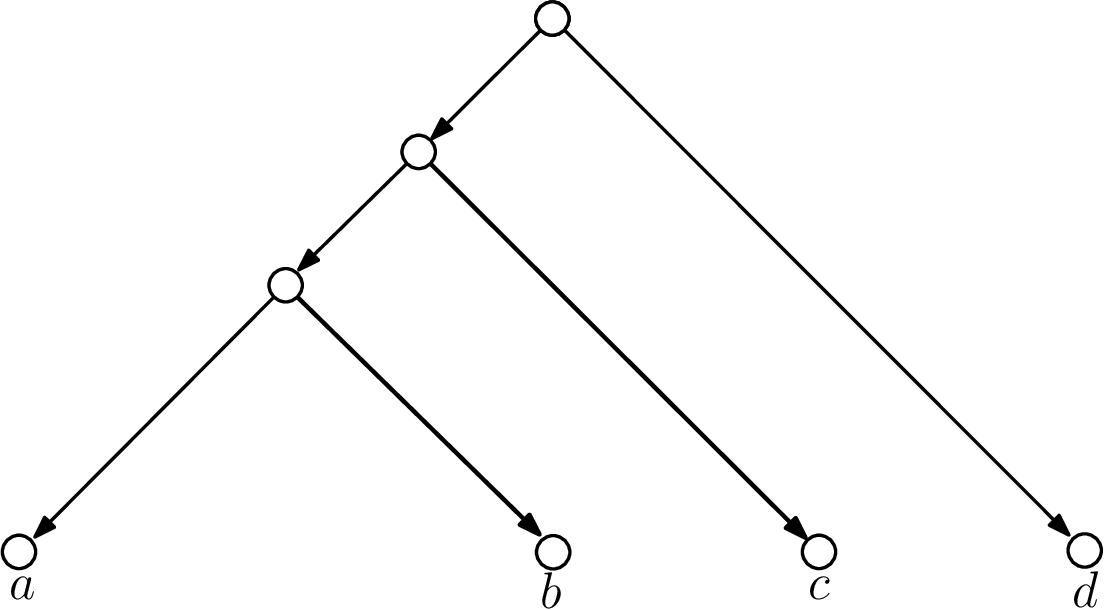}
        \caption{\(\T_1\)}
        \label{fig: example of phylognetic tree 1}
    \end{subfigure}
    \hfill
    \begin{subfigure}[b]{0.3\textwidth}
        \centering
        \includegraphics[width=\textwidth]{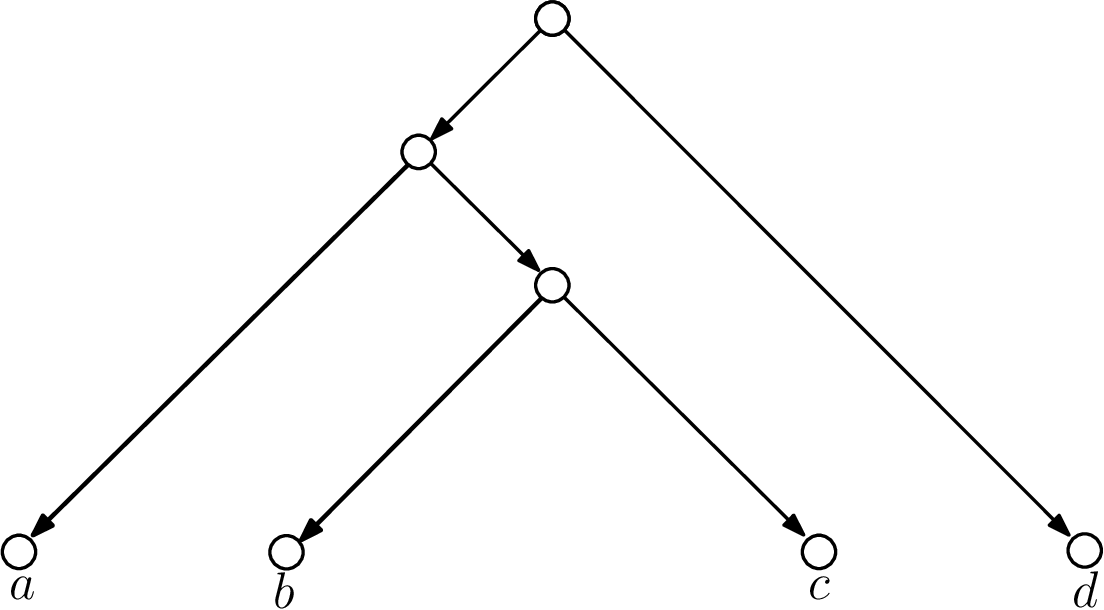}
        \caption{\(\T_2\)}
        \label{fig: example of phylognetic tree 2}
    \end{subfigure}
       \caption{(a) A rooted binary phylogenetic network \(\N\) on \(X = \{a, b, c, d\}\) with a single reticulation vertex. (b) and (c) The two rooted binary phylogenetic $X$-trees  $\T_1$ and $\T_2$ displayed by $\N$.}
      \label{fig: phylogenetic network and trees}
\end{figure}

Now, let \(\N\) be a phylogenetic network \blue{on $X$}, and let \(\T\) be a phylogenetic \(X\)-tree. We say \(\N\) \textit{displays} \(\T\) if \(\T\) can be obtained from \(\N\) by deleting vertices and edges, and by suppressing any resulting vertices of in-degree~1 and out-degree~1. Referring back to Figure \ref{fig: phylogenetic network and trees}, the phylogenetic network that is shown in (a) displays the two phylogenetic trees that are shown in (b) and (c).

We are now in a position to formally state \blue{the} Tree Containment \blue{decision problem}.\\

\noindent \textbf{Tree Containment} $(\N,\T)$  \\ 
\textbf{Input:} A phylogenetic \(X\)-tree \(\T\)and a phylogenetic network \(\N\) on \(X\). \\
\textbf{Output:} Does \(\N\) display \(\T\)?

\section{Quadratic Unconstrained Binary Optimization (QUBO)} \label{sec: QUBO}

In this section, we give a brief introduction to QUBO and discuss the general methodology to solve a problem using AQC. 

\subsection{What is QUBO?}\label{sec:QUBOdef}
QUBO is an NP-hard combinatorial optimization 
problem~\cite{glover2018tutorial,pardalos1992complexity}. Instances of many classical NP-hard problems such as finding a maximum cut or a minimum vertex coloring of a graph can be reduced to equivalent instances of QUBO. More specifically, QUBO  minimizes a quadratic objective function 
\(F: \mathbb{B}^n  \rightarrow \mathbb{R}\).
Using matrix notation, the quadratic objective function has the form \(H(\x) = \x^T Q \x \), where \(\x^T = [x_1,x_2, \dots, x_n]\) is a row vector of $n$ binary variables and $Q$ is and upper triangular \(n \times n\) matrix. Then the QUBO problem is that of solving the following equation
\begin{align*}
	x^{*} = \mathop{\text{min}}_{\x} \sum_{i = 1}^{n} \sum_{j = i}^{n} Q_{i,j} x_{i} x_{j} = \mathop{\text{min}}_{\x} \x^T Q \x,
\end{align*}
where the minimum is taken over all binary vectors \(\x\).  We  use \(x^{*}\)
to denote the minimum of  \(H\) and \(\x^{*}\) to denote a binary vector that
yields \(x^{*}\). In the quantum annealing model of  QUBO, the matrix \(Q\)
represents the problem Hamiltonian and each \(x_i\) in \(\x\) represents
a logical qubit. The logical qubits are different from the physical qubits
(qubits on a quantum annealer) as several physical qubits could be required
to represent a single logical qubit when we embed a given QUBO (non-zero
entries represent adjacency structure) onto the host
graph (physical qubits as vertices and couplers as edges) of a quantum annealer. The non-zero off-diagonal
entries, i.e. \(Q_{i,j}\) where \(i<j\), correspond to the coupler biases
between \(x_i\) and \(x_j\). Furthermore, the diagonal entries correspond
to the qubit biases, which \blue{refer} to the external magnetic fields applied
on the qubits.

\subsection{Methodology}

Suppose that we want to solve a problem \(P\) using the AQC model.  First we
need to establish a polynomial-time reduction that reduces a given instance
of \(P\) to an instance Q(\(P\)) of QUBO form.  Second, we need to ensure that the
\(n \times n\) matrix in Q(\(P\)) is as small and sparse as possible. This is
because the size and density of the QUBO matrix (more precisely, the density
of the graph whose weighted adjacency matrix is the QUBO matrix) have a
significant impact on the probability of the system being in the (minimum) ground
state in the final Hamiltonian.
We want to have a high enough probability
for the system to be in the ground state in the final Hamiltonian so that
we can efficiently query D-Wave's quantum annealers to solve Q(\(P\)).
Third, viewing the QUBO matrix in Q\((P)\) as a weighted adjacency matrix
of a graph \(G\), we (minor) embed this graph \(G\) onto the host graph of
a D-Wave's quantum annealer, i.e., either the Chimera graph (D-Wave 2000Q)
or the Pegasus graph (D-Wave Advantage). A {\it minor embedding} of a graph
\(G\) onto a graph \(H\) is a function \(\phi: V(G) \rightarrow 2^{V(H)}\)
that satisfies the following properties:

\begin{enumerate}
    \item The set of vertices $\phi(u)$ and $\phi(v)$ are disjoint for all $u,v \in V(G)$, where $u \neq  v$.
    \item For each $u \in V(G)$, there exists a subset $E^{\prime} \subset E(H)$ such that the subgraph $H^{\prime} = (\phi(u), E^{\prime})$ of $H$ is connected.
    \item For each $\{u, v\} \in E(G) $, there  exist vertices $u^{\prime} ,v^{\prime} \in V(H)$ such that $u^{\prime} \in \phi(u)$, $v^{\prime} \in \phi(v)$, and ${\{u' ,v^{\prime}\}} \in E(H)$.
\end{enumerate}
If  \(G\) is bigger than the host graph or if \(G\) cannot be embedded onto the host graph because it is too dense, then the package  \textit{qbsolv} that is provided by D-Wave can be used to break the \(n \times n\) matrix in Q(\(P\)) into sub-matrices and solve them separately before combining the results to get a solution for the initial problem. This package uses techniques from well-known paradigms such as divide-and-conquer and dynamic programming; for more information, see \cite{D-Wave_qbsolv}. In the last step, a quantum annealer is queried to compute \(x^*\) and \(\x^*\). 

A problem one might immediately see is that finding a minor embedding of
a graph $G$ onto a host graph is an NP-hard problem in itself. 
Furthermore, if we can find an embedding there are often better ones 
(e.g.~those that minimize the maximum or average cardinality of $\phi(u)$);
here the chance of the quantum annealer successfully solving \blue{Q$(P)$}
increases with better embeddings due to hardware limitations.
The extended optimization problem of finding an embedding with 
maximum mapping size at most $k$ \blue{(i.e. $|\phi(u)|\leq k$ for each $u\in V(G)$)}) is also NP-hard. \blue{Note that, if} we fix $k=1$, then we solve the NP-hard subgraph isomorphism problem and if we fix $k=|\blue{V(H)}|$, then we
solve the original minor-embedding problem.
However, as it is not necessary to find an optimal embedding with respect
to some criteria, we can use a probabilistic algorithm (on a classical
computer) with a polynomial-time overhead to find a `\blue{good}' minor embedding of $G$.

After getting the results from a quantum annealer, we need a way to decode the input values \(\x^*\) that yield
\(x^{*}\) for an instance of a problem \(P\). In practice, the final
solution is probabilistic. Since we might not get the optimal answer the
first time that we query a quantum annealer, we may need to query it several times (often more than 100 times, in practice) to increase the probability of getting the optimal answer. 
If \(P\) is a decision problem in NP, as it is in the case of Tree Containment, then it is often possible to reduce \(P\) to Q(\(P\)) such that the optimal value after post-processing is 0 if the answer to $P$ is `yes'. 
Post-processing is the process of adding an offset to the optimal value \(x^*\). Using this approach, we know when we get the optimal solution from a quantum annealer. We use this approach in our reduction from Tree Containment to QUBO in Section~\ref{sec: TC problem}. 

\subsection{Reducing Higher-Order Functions into QUBO}

Some problems naturally reduce to a binary cubic, or \blue{binary} higher-order
function. Although such a function does not immediately fit into the QUBO
framework, it can be recast as a binary quadratic function and then be solved
with D-Wave's quantum annealers. We introduce additional binary variables
during the recast and replace the higher-order terms with \blue{additional} penalty functions
that have binary quadratic order terms.
See~\cite{calude2017solving} for a detailed example of converting a traditional Integer Linear Programming
formulation for a constrained optimization problem to QUBO form, which
uses auxiliary variables for reducing higher-order \blue{functions} and models
non-binary variables as sets of binary variables.

The following lemma and proposition
reduces a binary cubic function to a binary quadratic function. The lemma
has been taken from~\cite{glover2018tutorial}, where no proof is given. For
reasons of completeness, we next establish a formal proof.

\begin{lemma}\label{lem: reduction technique}
Let $x_1$, $x_2$, and $y$ be three binary variables. Furthermore, let $P = x_1 x_2 - 2x_1 y -2 x_2 y + 3y$. Then  \( P = 0\) if and only if \(y = x_1 x_2\). 
\end{lemma}

\begin{proof}
	\((\implies)\) Suppose that $P=0$. We consider four cases.
	\begin{enumerate}
		\item If \(x_1 = 1, x_2 = 1\) and \(P=0\), then \(0 = 1 - 2y -2y + 3y \) and hence \( y = 1\).
		\item If \(x_1 = 1, x_2 = 0\) and \(P=0\), then \(0 = 0 -2y +3y \) and hence \(y = 0\).
		\item If \(x_1 = 0, x_2 = 1\) and \(P=0\), then \(0 = 0 -2y +3y \) and hence \( y = 0\).
		\item If \(x_1 = 0, x_2 = 0\) and \(P=0\), then \(0 = 0 + 3y \) and hence \( y = 0\).
	\end{enumerate}
	From each case, it follows that \(y = x_1 x_2\).
	
	\((\impliedby)\) Now suppose that \(y = x_1 x_2\). Then, since $x^2 =x$ for any binary variable, we have $$P = x_1 x_2 - 2x_1 x_1 x_2 -2 x_2 x_1 x_2 + 3x_1 x_2 = x_1 x_2 - 2x_1 x_2 -2x_1 x_2 + 3x_1 x_2  =  0.$$
The lemma follows.
\end{proof}

\begin{proposition} \label{prop: reduction technique}
	\blue{In a QUBO framework, a} binary cubic \blue{term} can be converted into \blue{an equivalent} binary quadratic  \blue{term}.
\end{proposition}
\begin{proof}
   Let $x_1x_2x_3$ be a binary cubic term. Furthermore, let \(y\) be a new binary variable, and let $$P = x_1 x_2 - 2x_1 y -2 x_2 y + 3y$$ be a binary quadratic penalty term. It follows by Lemma~\ref{lem: reduction technique} that  \( P = 0\) if and only if \(y = x_1 x_2\). Hence, substituting $x_1x_2x_3$ with $x_3y+P$ replaces a binary cubic term with \blue{four} binary quadratic terms \blue{and one binary linear term}. Specifically, if we minimize  \(P\) when we minimize the resulting  quadratic function, then \(y\)  represent \(x_1 x_2\). Applying this substitution technique repeatedly to each binary cubic term of a binary cubic function results in a binary quadratic function.  	
\end{proof}

\noindent Note that we can recursively apply Proposition~\ref{prop: reduction technique} to reduce any binary higher-order function to a binary quadratic function.

\section{Reducing Tree Containment to QUBO} \label{sec: TC problem}
In this section, we present a reduction from Tree Containment to QUBO. For an instance of Tree Containment that consists of a phylogenetic network $\N$ on $X$ and a phylogenetic $X$-tree $\T$, the QUBO formulation requires $O(n_{\T} n_{\N})$ logical qubits, where \(n_{\N}\) is the number of vertices in  \(\N\) and \(n_{\T}\) is the number of vertices in \(\T\). 

\subsection{QUBO Formulation} \label{sec: QUBO Formulation}

Throughout this section, let $\N$ be a phylogenetic network on $X$, and let $\T$ be a phylogenetic $X$-tree. Let \(E(\N)\) be the set of edges and \(V(\N)=\{v_0,v_1,v_2,\ldots,v_{n_\N-1}\}\) be the set of vertices of \(\N\), and let \(E(\T)\) be the set of edges and \(V(\T)=\{u_0,u_1,u_2,\ldots,u_{n_\T-1}\}\) be the set of vertices of \(\T\). Without loss of generality, we may assume that $u_0$ is the root of $\T$. Additionally, let $u_{n_\T}$ be a vertex that is not an element in $V(\T)$. 

Intuitively, if $\N$ displays $\T$, then there exists a mapping that maps each vertex of $\T$ to a vertex of $\N$ and each edge of $\T$ to a directed path of $\N$. The following QUBO formulation for Tree Containment establishes a mapping (detailed below) that maps each vertex of 
$V(\T)$ 
to at least one vertex of $\N$. In this mapping, $u_{n_\T}$ is mapped to each vertex of $\N$ that is not in the image of any vertex in $V(\T)$. 

Now, let $$m=n_{\N}(n_{\T}+1) + (n_{\T} - 1 )\left(1 + \left\lfloor \lg(n_{\N}-n_{\T})\right\rfloor\right) + n_\T( \alpha + \beta)+ 2 \beta \gamma,$$ and let $\x \in \mathbb{B}^m$ be a vector of binary variables, where \(n_{\N} = \left\lvert V(\N) \right\rvert \), \(n_{\T} = \left\lvert V(\T) \right\rvert\), \(\gamma = n_\T - |X|\), and \(\alpha\) (resp. $\beta$) equals the number of reticulation vertices (resp. tree vertices) of \(\N\). 
It immediately follows that $\x$ contains $O(n_\N n_\T)$ binary variables.

We next describe the binary variables that are represented by $\x$ and
their encoding. More precisely, for \(0 \leq i \leq n_{\T}\)  and \( 0 \leq j < n_{\N} \), \( x_{i,j} = 1\) encodes that \( u_i \in V(\T) \cup \{u_{n_{\T}}\}\) is mapped to \( v_j \in V(\N)\) and, similarly, \( x_{i,j} = 0\) encodes that \( u_i \in V(\T) \cup \{u_{n_{\T}}\}\) is not mapped to \( v_j \in V(\N)\). Additionally, we introduce three types of slack variables. 
\begin{enumerate}
\item For each vertex \( u_i \in V(\T)\backslash\{u_0\}\), we have \(1 + \left\lfloor \lg(n_{\N}-n_{\T}\right)\rfloor \) slack variables that are denoted by \(y_{i,r}\) for \(0 \leq r \leq \left\lfloor \lg(n_{\N} -n_{\T})\right\rfloor \).

\item For each \(u_i \in V(\T)\), we have \(\alpha + \beta\) slack variable that are denoted by \(z_{i,j}\) for each index $j$ such that \(v_j \in V(\N)\) is either a tree or reticulation vertex. 
\item For each \(u_i \in V(\T)\) that is not a leaf, we have \(2\beta\) slack variables that are denoted by \(\hat{z}_{i,2j}\) and \(\hat{z}_{i,2j+1}\) for each index $j$ such that \(v_j \in V(\N)\) is a tree vertex.
\end{enumerate}

For the following Hamiltonian, we assume without loss of generality that \( |V(\N)| \geq |V(\T)|\). Indeed, if this is not the case, then \(\N\) does not display \(\T\). Let $v_j$ be a tree vertex of $\N$, and let \(v_{\childIndex{1}{j}}\) and \(v_{\childIndex{2}{j}}\) be the two children of  \(v_j\), where $j_1$  and $j_2$ are the indices of the children of \(v_j\). We use \(\child{1}{v_j}\) and \(\child{2}{v_j}\) to denote \(v_{\childIndex{1}{j}}\) and \(v_{\childIndex{2}{j}}\), respectively. Now, let $v_j$ be a reticulation vertex of $\N$. Similarly to the children of a tree vertex, let \(v_{\parentIndex{1}{j}}\) and \(v_{\parentIndex{2}{j}}\) be the two parents of  \(v_j\), where \(\parentIndex{1}{j}\) and \(\parentIndex{2}{j}\) are the indices of the two parents of \(v_j\). Again, we use \(\parent{1}{v_j}\) and \(\parent{2}{v_j}\) to denote \(v_{\parentIndex{1}{j}}\) and \(v_{\parentIndex{2}{j}}\), respectively. Furthermore, we define two functions \(f\) and \(g\) as follows.

\begin{align*}
    f(u_i,u_l) & = \begin{cases}
        1, \text{if there exists an edge from \(u_i\) to \(u_l\) in \(\T\)} \\
        0, \text{otherwise}
    \end{cases}\\ 
    g(v_j,v_k) & = \begin{cases}
        1, \text{if there exists an edge from \(v_j\) to \(v_k\) in \(\N\)} \\
        0, \text{otherwise}
    \end{cases} \\
\end{align*}

We are now in a position to define the Hamiltonian $H(\x)$, also sometimes called the objective function, as follows, where $A,B\in\mathbb{R}^+$. The choice of $A$ and $B$ is detailed in Section~\ref{sec:main-result}. However, we already note here that $B$ is sufficiently larger than $A$.
The coefficients of the terms $x_ix_j$ of the following binary objective
function $H(\x)$ correspond to the entries in the QUBO matrix $Q$, as
defined in Section~\ref{sec:QUBOdef}.

\begin{align*}
    H(\x) & =  B \cdot  \Biggr( \sum_{I = 1}^{10}   P_I(\x) \Biggl) +  A \cdot P_{11}(\x) + P_{12}(\x) 
    \shortintertext{where} \\
    P_1(\x) & = \Biggl( 1 - \sum_{j = 0}^{n_{\N}-1} x_{0,j} \Biggr)^2 + \sum_{i = 1}^{n_{\T}-1} \Biggl( 1 - \sum_{j = 0}^{n_{\N}-1} x_{i,j}  + \sum_{r = 0}^{\lfloor \lg(n_{\N} -n_{\T})\rfloor  } 2^r y_{i,r} \Biggr)^2 \\
    P_2(\x) & = \sum_{j = 0}^{n_{\N}-1} \Biggl( \sum_{i = 0}^{n_{\T}} x_{i,j} - 1  \Biggr)^2 \\
    P_3(\x) & = \sum_{i=0}^{n_{\T}-1}  \sum_{ \substack {v_j \text{ is a tree } \\ \text{ vertex of \(\N\) } }}^{} \Biggr( x_{i, \childIndex{1}{j}} x_{i, \childIndex{2}{j}} -2 x_{i, \childIndex{1}{j}} z_{i,j} -2 x_{i, \childIndex{2}{j}} z_{i,j} + 3 z_{i,j}  \Biggr) \\
    P_4(\x) & = \sum_{i=0}^{n_{\T}-1}  \sum_{ \substack {v_j \text{ is a tree } \\ \text{ vertex of \(\N\) } }}^{} x_{i,j} z_{i,j} \\
    P_5(\x) & = \sum_{i=0}^{n_{\T}-1}  \sum_{ \substack {v_j \text{ is a reticulation } \\ \text{ vertex of \(\N\) } }}^{} \Biggr( x_{i, \parentIndex{1}{j}}  x_{i, \parentIndex{2}{j}} -2 x_{i, \parentIndex{1}{j}} z_{i,j} -2 x_{i, \parentIndex{2}{j}} z_{i,j} + 3 z_{i,j} \Biggl) \\
    P_6(\x) & =  \sum_{i=0}^{n_{\T}-1}  \sum_{ \substack {v_j \text{ is a reticulation } \\ \text{ vertex of \(\N\) } }}^{} x_{i,j} z_{i,j} \\
    P_7(\x) & = \sum_{i = 0}^{n_\T-1} \sum_{\substack{l = 0 \\ l \neq i}}^{n_\T-1} \Biggl( f(u_i, u_l) \Biggl(  \sum_{ \substack {v_j \text{ is a tree } \\ \text{ vertex of \(\N\) } }}^{} x_{i,j} z_{l,j} \Biggr) \Biggr) \\
    P_8(\x) & = \sum_{ \substack {u_i \text{ is not a} \\ \text{ leaf of \(\T\) } } } \sum_{ \substack {v_j \text{ is a tree } \\ \text{ vertex of \(\N\) } }}^{} \Biggl( x_{i,j} x_{i, \childIndex{1}{j}} - 2 x_{i,j} \hat{z}_{i,2j} - 2 x_{i, \childIndex{1}{j}} \hat{z}_{i,2j} + 3 \hat{z}_{i,2j} 
    \mspace{150mu}
    \notag\\
    & \;\;\;\;\;\;\;\;\;\;\;\;\;\;\;\;\;\;\;\;\;\;\;\;\;\;\;\;\;\;\;\;\;\;\;
    + x_{i,j} x_{i, \childIndex{2}{j}} - 2 x_{i,j} \hat{z}_{i,2j+1} - 2 x_{i, \childIndex{2}{j}} \hat{z}_{i,2j+1} + 3 \hat{z}_{i,2j+1} \Biggr)   \\
    P_9(\x) & =  \sum_{i=0}^{n_{\T}-1}  \sum_{ \substack {l=0 \\ l \neq i}}^{n_{\T}-1} \Biggl( f(u_i, u_l) \Biggl(  \sum_{ \substack {v_j \text{ is a tree } \\ \text{ vertex of \(\N\) } }}^{} \Biggl( \hat{z}_{i,2j} x_{l, \childIndex{2}{j}}+ \hat{z}_{i, 2j+1} x_{l, \childIndex{1}{j}}\Biggr) \Biggr) \Biggr)  \\
       P_{10}(\x) & = \sum_{ \substack{ u_i \text{ is a } \\ \text{leaf of  } \T }}^{} \Biggl( 1 -  x_{i,j} \Biggr)^2  \text{\; , where \(j\) is the index of \(u_i\) in \(\N\)}\\
    P_{11}(\x) & = \sum_{i=0}^{n_{\T}-1}  \sum_{j= 0}^{n_{\N}-1} \Biggl( x_{i,j} \Biggl(1 - \sum_{\substack{k=0 \\ k \neq j }}^{n_{\N}-1} g(v_j, v_k) x_{i,k} \Biggr)   \Biggr)   - n_{\T}  \\
    P_{12}(\x) & = \sum_{i=0}^{n_{\T}-1}  \sum_{\substack{l=0 \\ l \neq i}}^{n_{\T}-1} \Biggl( f(u_i, u_l) \Biggl( 1 - \sum_{j= 0}^{n_{\N}-1} \sum_{ \substack{ k=0 \\ k \neq j}}^{n_{\N}-1} g(v_j, v_k)  x_{i,j} x_{l,k}  \Biggr) \Biggr)  \\
\end{align*} \\

For notational convenience, we use a map $d$ that maps a tuple
$(\x,u_i)$ to a subset of $V(\N)$. More precisely, we define \[
\D: (\blue{\mathbb{B}^m},
\allowbreak V(\T) \cup \{u_{n_\T}\})\rightarrow 2^{V(\N)} \] to decode
the subset of vertices \blue{of $\N$} such that \( \D( \x, u_i )= \{v_j\in V(\N) \mid x_{i,
j}  = 1\}\). We interpret this as the vertex \(u_i\) of $\T$ being mapped
to the subset  \(\{v_j\in V(\N) \mid x_{i, j}  = 1\}\) of $V(\N)$. We
also sometimes interpret \(\D(\x, u_i)\) as an induced subgraph of \(\N\) \blue{whose vertex set is} $ \D(\x, u_i)$
\blue{and whose edge set} is $\{(v_c, v_d) \in E(\N) \mid v_c,
v_d \in \D(\x, u_i)\}$.

We now return to the Hamiltonian $H(\x)$ as defined above and establish a lemma for each of $P_1(\x),P_2(\x),\ldots,P_{12}(\x)$. These lemmas provide some insight into different parts of $H(\x)$. Lemmas~\ref{P3}, \ref{P5}, and \ref{P8} follow from the proof of Lemma~\ref{lem: reduction technique}, whereas the proofs  of Lemmas~\ref{P1}, \ref{P2}, and \ref{P10} are straightforward and omitted.

The first penalty \blue{function} $P_1$ ensures that the root of \blue{$\T$} is mapped to \blue{exactly} one vertex of \blue{$\N$} and each other vertex of $\T$ is mapped to at least one vertex of $\N$.

\begin{lemma}\label{P1}
\(P_1(\x) = 0 \) if and only if, for each vertex \(u_i \in V(\T)\), we have \(|\D(\x,u_i)| > 0\) and \(|\D(\x, u_0)| = 1\).
\end{lemma}

Next we ensure that at most one vertex of \blue{$\T$} is mapped to each vertex of \blue{$\N$}.

\begin{lemma}\label{P2}
\(P_2(\x) = 0\) if and only if, for each vertex \(v_j \in V(\N)\), there exists exactly one vertex \(u_i\in V(\T)\cup\{u_{n_\T}\}\)  such that \( v_j \in \D(\x, u_i)\).
\end{lemma}

The penalty \blue{function $P_3$} establishes \blue{equivalence between slack variables and the product of two non-slack variables}.

\begin{lemma}\label{P3}
\(P_3(\x) = 0\) if and only if \(z_{i,j} = x_{i, \childIndex{1}{j}} x_{i, \childIndex{2}{j}} \), where \(v_j \in V(\N)\) is a tree vertex and \(u_i \in V(\T)\). 
\end{lemma}

We now require that \blue{no} vertex of $\T$ \blue{maps to a tree vertex in $\N$ and its two children}.

\begin{lemma}\label{P4}
Suppose that \(P_3(\x) = 0\). Then \(P_4(\x) = 0\) if and only if there does not exist a tree vertex \(v_j \in V(\N)\) and a vertex $u_i\in V(\T)$ such that $\{v_j,c_1(v_j),c_2(v_j)\}\subseteq d(\x,u_i)$.
\end{lemma}

\begin{proof}
    As \(P_3(\x) = 0\), it follows from Lemma~\ref{P3} that $z_{i,j}=x_{i,j_1}x_{i,j_2}$. Thus, \(P_4(\x) = 0\) if and only if 
    \begin{align}
        \label{eq: P3 and P4}
        \sum_{i=0}^{n_{\T}-1}  \sum_{ \substack {v_j \text{ is a tree } \\ \text{ vertex of \(\N\) } }}^{} x_{i,j} x_{i, \childIndex{1}{j}} x_{i, \childIndex{2}{j}} = 0 
    \end{align}
It now follows that Equation (\ref{eq: P3 and P4}) holds if and only if there exists no tree vertex $v_j\in V(\N)$ such that  \(\{v_j, \child{1}{v_j}, \child{2}{v_j}\} \subseteq \D(\x, u_i)\) for some \(u_i \in V(\T\)). 
\end{proof}

\blue{The penalty function $P_5$ is similar to $P_3$.}

\begin{lemma}\label{P5}
\(P_5(\x) = 0\) if and only if \(z_{i,j} = x_{i, \parentIndex{1}{j}} x_{i, \parentIndex{2}{j}} \), where \(v_j \in V(\N)\) is a reticulation vertex and \(u_i \in V(\T)\).
\end{lemma}

We next require that \blue{no} vertex of $\T$ \blue{maps to a reticulation vertex of $\N$ and its two parents.}

\begin{lemma}\label{P6}
Suppose that \(P_5(\x) = 0\). Then \(P_6(\x) = 0\) if and only if there does not exist a reticulation vertex \(v_j \in V(\N)\) and a vertex $u_i\in V(\T)$ such that $\{v_j,p_1(v_j),p_2(v_j)\}\subseteq d(\x,u_i)$.
\end{lemma}

\begin{proof}
    The proof is analogous to that of Lemma~\ref{P4}. As \(P_5(\x) = 0\), it follows from Lemma~\ref{P5} that $z_{i,j}=x_{i,j^1}x_{i,j^2}$. Thus, \(P_6(\x) = 0\) if and only if 
    \begin{align}
        \label{eq: P5 and P6}
        \sum_{i=0}^{n_{\T}-1}  \sum_{ \substack {v_j \text{ is a reticulation } \\ \text{ vertex of \(\N\) } }}^{} x_{i,j} x_{i, \parentIndex{1}{j}} x_{i, \parentIndex{2}{j}} = 0 
    \end{align}
It now follows that Equation (\ref{eq: P5 and P6}) holds if and only if there exists no reticulation vertex $v_j\in V(\N)$ such that \(\{v_j, \parent{1}{v_j}, \parent{2}{v_j}\} \subseteq \D(\x, u_i)\) for some \(u_i \in V(\T\)). 
\end{proof}

Furthermore, \blue{we ensure that there exists no edge $(u_i,u_l)$ in $\T$ such that $u_i$ maps to a tree vertex $v_j$ of $\N$ and $u_l$ maps to both children of $v_j$.}

\begin{lemma}\label{P7} 
Suppose that \(P_3(\x) = 0\). Then \(P_7(\x) = 0\) if and only if there does not exist a tree vertex \(v_j \in V(\N)\) and an edge $(u_i,u_l)\in E(\T)$ such that \(v_j \in \D(\x, u_i)\) and \(\{\child{1}{v_j}, \child{2}{v_j}\}\subseteq \D(\x, u_l)\).
\end{lemma}

\begin{proof}
    As \(P_3(\x) = 0\) it again follows from Lemma ~\ref{P3} that $z_{l,j}=x_{l,j_1}x_{l,j_2}$, where $l=i$. Hence, \( P_7(\x) = 0\) if and only if 
    \begin{align}
        \label{eq: P3 and P7}
        \sum_{i = 0}^{n_\T-1} \sum_{\substack{l = 0 \\ l \neq i}}^{n_\T-1} \Biggl( f(u_i, u_l) \Biggl(  \sum_{ \substack {v_j \text{ is a tree } \\ \text{ vertex of \(\N\) } }}^{} x_{i,j} x_{l, \childIndex{1}{j}}x_{l, \childIndex{2}{j}}\Biggr) \Biggr) = 0
    \end{align}
The lemma now follows from Equation~\ref{eq: P3 and P7}.
\end{proof}

The next penalty \blue{function} $P_8$ \blue{introduces additional slack variables to avoid binary cubic terms in $P_9$.}

\begin{lemma}\label{P8}
\(P_8(\x) = 0\) if and only if \(\hat{z}_{i,2j} = x_{i,j} x_{i, \childIndex{1}{j}}\) and \(\hat{z}_{i,2j+1} = x_{i,j} x_{i, \childIndex{2}{j}}\), where \(u_i \in V(\T)\blue{\setminus X}\) and \(v_j \in V(\N)\) is a tree vertex.
\end{lemma}

\blue{The next lemma shows that two vertices that are incident with a given edge in $\T$ are not mapped to two distinct vertices in $\N$ that have a common parent.} 

\begin{lemma}\label{P9}
Suppose that \(P_8(\x) = 0\). Then \(P_9(\x) = 0\) if and only if, there does not exist a tree vertex \(v_j \in V(\N)\) and an edge $(u_i,u_l)\in E(\T)$ such that \(\{v_j, \child{1}{v_j}\}\subseteq \D(\x, u_i)\) and \(\child{2}{v_j} \in \D(\x, u_l)\), or \(\{v_j, \child{2}{v_j}\}\subseteq \D(\x, u_i)\) and \(\child{1}{v_j} \in \D(\x, u_l)\).
\end{lemma}

\begin{proof}
    Since \(P_8(\x) = 0\), it follows from Lemma~\ref{P8} that \( P_9(\x) = 0\) if and only if 
    \begin{align}
        \label{eq: P8 and P9}
        \sum_{i=0}^{n_{\T}-1}  \sum_{ \substack {l=0 \\ l \neq i}}^{n_{\T}-1} \Biggl( f(u_i, u_l) \Biggl(  \sum_{ \substack {v_j \text{ is a tree } \\ \text{ vertex of \(\N\) } }}^{} \Biggl( x_{i,j} x_{i, \childIndex{1}{j}} x_{l, \childIndex{2}{j}}+ x_{i,j} x_{i, \childIndex{2}{j}} x_{l, \childIndex{1}{j}}\Biggr) \Biggr) \Biggr) = 0
    \end{align}
    The lemma now follows from Equation~\ref{eq: P8 and P9}.
\end{proof}

The penalty \blue{function} $P_{10}$ ensures that the leaf labels $X$ match in both $\T$ and $\N$.

\begin{lemma}\label{P10}
 \(P_{10}(\x) = 0\) if and only if, for each leaf $u_i$ of $\T$, there exists a leaf $v_j$ in $\N$ such that $u_i$ and $v_j$ have the same label and \( v_j \in \D(\x, u_i) \).
\end{lemma}

We now restrict the mapping of \blue{each} vertex of $\T$ to induce a \blue{directed} path in $\N$.

\begin{lemma}\label{P11}
Suppose that  \(P_I(\x) = 0\) for each \( 3 \leq I \leq 6\). Then \(P_{11}(\x) = 0\) if and only if, for each vertex \(u_i \in V(\T)\), \(\D(\x, u_i)\) is a directed path of \(\N\).
\end{lemma}

\begin{proof}
    We first notice that \(P_{11}(\x)\) adds a penalty if and only if, for two vertices \(u_i \in V(\T)\) and \(v_j \in V(\N)\), we have $v_j\in\D(\x,u_i)$ but no child of $v_j$ in $\N$ is contained in \(\D(\x,u_i)\). Since there is no directed cycle in \(\N\), \(P_{11}(\x)\) adds at least a penalty of  1 for each vertex in \(V(\T)\). 
    
    (\(\implies\)) Suppose that \(P_{11}(\x) = 0\). Towards a contradiction, assume that there exists a vertex $u_a\in V(\T)$ such that \(\D(\x, u_a)\) does not form a directed path in \(\N\). This implies that there exists a vertex \(v \in \D(\x, u_a)\) such that $v$ has two parents that are both contained in \(\D(\x, u_a)\), $v$ has two children that are both contained in \(\D(\x, u_a)\), or \(\D(\x,u_a)\) is disconnected. Since $P_3(\x) = P_4(\x) =P_5(\x) = P_6(\x)=0$, it follows from Lemmas~\ref{P3}--\ref{P6}, that each vertex in \(\D(\x, u_a)\) has at most one child in $\N$ that is contained in \(\D(\x, u_a)\) and at most one parent in $\N$ that is contained in \(\D(\x, u_a)\). Thus, \(\D(\x,u_a)\) is disconnected, and there exist \(v_c, v_d \in \D(\x, u_a)\) such that no vertex of $\N$ that is a child of  \(v_c\) or \(v_d\) is contained in  \(\D(\x, u_a)\). Then $$\sum_{\substack{k=0 \\ k \neq c }}^{n_{\N}-1}  g(v_c, v_k)x_{a,k} = 0 \text{ and }\sum_{\substack{k=0 \\ k \neq d }}^{n_{\N}-1} g(v_d, v_k) x_{a,k} = 0.$$ 
Therefore, \(u_a\) adds a penalty of at least 2 and every  vertex in \(V(\T)\backslash\{u_a\}\) adds a penalty of at least 1. As we subtract \(n_\T\) in $P_{11}(\x)$ it follows that \(P_{11}(\x) > 0\); a contradiction. Hence, \blue{if $P_{11}(\x)=0$, then} \(\D(\x,u_i)\) is a directed path of $\N$ for each $u_i \in V(\T)$. 
    
   ( \( \impliedby\)) Let $u_i\in V(\T)$. Suppose that \(\D(\x, u_i)\) is a directed path of \(\N\). As $\N$ is acyclic, there exists exactly one vertex \(v \in \D(\x,u_i)\) such that no child of \(v\) in $\N$ is contained in \(\D(\x,u_i)\). This implies that \(u_i\) adds a penalty of 1. In total, we have \(n_\T\) vertices in \(\T\) and, so a total penalty of \(n_\T\). Since we subtract \(n_\T\) in $P_{11}(\x)$ it follows that \(P_{11}(\x) = 0\). \blue{The lemma now follows.}
\end{proof}

Finally we restrict the two \blue{directed} paths in $\N$ that correspond to two \blue{adjacent} vertices of $\T$ \blue{to be separated by} exactly one edge in $\N$.  That is, \blue{the subgraph of $\N$ that is induced collectively by all vertices in $\T$} is a tree.

\begin{lemma}\label{P12}
Suppose that \(P_I(\x) = 0\) for each \(3 \le I \le 9\) and $P_{11}(\x) = 0$. Then \(P_{12}(\x) = 0\) if and only if, for each edge \((u_i,u_l) \in E(\T)\), there exists exactly one edge \((v_c, v_d) \in E(\N)\) such that \(v_c \in\D(\x, u_i)\) and \(v_d \in \D(\x, u_l)\).
\end{lemma}

\begin{proof}
    Let $(u_i,u_l)\in E(\T)$. We start by noticing that \(P_{12}(\x)\) adds a penalty if and only if there does not exist an edge from a vertex in \(\D(\x, u_i)\) to a vertex in \(\D(\x, u_l)\) in \(\N\). 
    
     (\(\implies\)) Suppose that \(P_{12}(\x) = 0\). 
    Towards a contradiction, assume that there exist at least two edges \((v_c, v_d), (v_{c'},v_{d'}) \in E(\N)\) such that $\{v_c,v_{c'}\}\subseteq \D(\x,u_i)$ and $\{v_d,v_{d'}\}\subseteq \D(\x,u_l)$. By Lemma~\ref{P11}, \(\D(\x,u_i)\) and \(\D(\x,u_l)\) are two directed paths of \(\N\). We consider two cases.  
    
\noindent \textbf{Case 1.} Assume that $v_c\ne v_{c'}$. Then, without loss of generality, we may assume that $v_{c'}$ precedes $v_c$ on the directed path $\D(\x,u_i)$. Hence, $v_{c'}$ is a tree vertex of $\N$. Furthermore, we have \(\{v_{c'}, \child{a}{v_{c'}}\}\subseteq  \D(\x,u_i)\) and \(\child{b}{v_{c'}} \in \D(\x,u_l)\), where \( \{a,b\}=\{1,2\}\). This setup is shown in Figure \ref{fig: tree vertex between two chain graphs}. As \(P_8(\x)=0\), it follows by Lemma~\ref{P9} that  \(P_9(\x)>0\); a contradiction. 
    \begin{figure}[h] 
        \centering
        \includegraphics[width=6cm]{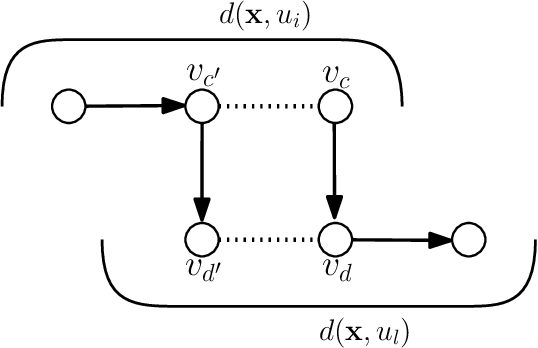}
        \caption{Setup as described in Case 1 of the proof of Lemma~\ref{P12}. The vertices of the top directed path and the bottom directed path represent \(\D(\x, u_i)\) and \(\D(\x, u_l)\), respectively. The tree vertex $v_{c'}$  has one child in $d(\x,u_i)$ and the other child in $d(\x,u_l)$.}
        \label{fig: tree vertex between two chain graphs}
    \end{figure}
    
\noindent \textbf{Case 2.} Assume that $v_c= v_{c'}$. It follows that $v_c$ is a tree vertex of $\N$, $v_c\in d(\x,u_i)$, and \(\{\child{1}{v_c}, \child{2}{v_c}\}\subseteq \D(\x, u_l)\). This setup is shown in Figure~\ref{fig: tree head vertex has two edges}. Now, as \(P_3(\x) = 0\), Lemma~\ref{P7} implies that \(P_7(\x) > 0\); again a contradiction. 
    \begin{figure}[h] 
        \centering
        \includegraphics[width=7.3cm]{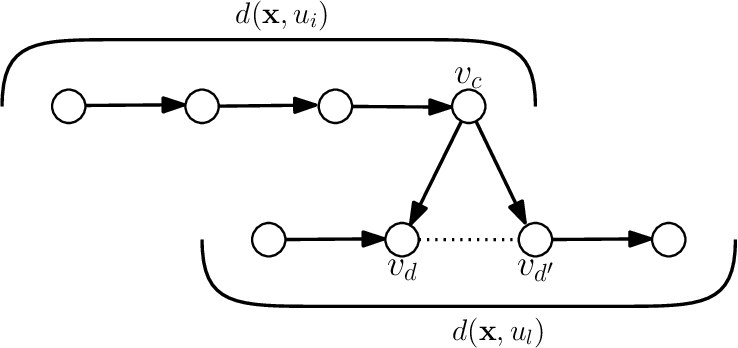}
        \caption{Setup as described in Case 2 of the proof of Lemma~\ref{P12}. The vertices of the top directed path and the bottom directed path represent  \(\D(\x, u_i)\) and \(\D(\x, u_l)\), respectively. The tree vertex $v_c$ has both children in $d(\x,u_l)$.} 
        \label{fig: tree head vertex has two edges}
    \end{figure} 
    
    By combining both cases, there exists at most one edge from a vertex in \(\D(\x,u_i)\) to a vertex in \(\D(\x,u_l)\) in \(\N\). Thus $$\sum_{j= 0}^{n_{\N}-1} \sum_{ \substack{ k=0 \\ k \neq j}}^{n_{\N}-1} g(v_j, v_k) x_{i,j} x_{l,k}   \leq 1.$$ Moreover, as \(P_{12}(\x) = 0\), we have $$\sum_{j= 0}^{n_{\N}-1} \sum_{ \substack{ k=0 \\ k \neq j}}^{n_{\N}-1} g(v_j, v_k) x_{i,j} x_{l,k}   = 1.$$ Hence, for \((u_i,u_l)\) there exists exactly one edge in $\N$ that has one endpoint in \(\D(\x, u_i)\) and the other endpoint in \( \D(\x, u_l)\).

    (\(\impliedby\)) Suppose that, for each edge \((u_i,u_l) \in E(\T)\), there exists exactly one edge \((v_c, v_d) \in E(\N)\) such that \(v_c \in\D(\x, u_a)\) and \(v_d \in \D(\x, u_b)\). Then, $$\sum_{j= 0}^{n_{\N}-1} \sum_{ \substack{ k=0 \\ k \neq j}}^{n_{\N}-1} g(v_j, v_k) x_{i,j} x_{l,k}   = 1$$ in $P_{12}(\x)$ and, so,  \(P_{12}(\x) = 0\). \\
\end{proof}

The next corollary follows from Lemmas~\ref{P9} and \ref{P12}.

\begin{corollary}\label{cor:connect}
Suppose that \(P_I(\x) = 0\) for all \( I \in \{1,2,\ldots,9,11,12\}\). Let \((u_i, u_l) \in E(\T)\), then  there  exists an edge from the terminal vertex of the directed path induced by  \(\D(\x, u_i)\) to a vertex of the directed path induced by \(\D(\x, u_l)\) in \(\N\).
\end{corollary}

\subsection{Proof of Correctness} \label{sec:main-result}

In this section, we show that the Hamiltonian $H(\x)$ as presented in Section~\ref{sec: QUBO Formulation} correctly encodes instances of Tree Containment. We start by detailing the choice of constants $A$ and $B$ in the definition of $H(\x)$.

\begin{lemma} \label{lemma: min value of P12}
    Let \(\x \in \mathbb{B}^m\). If $P_I(\x) = 0$ for all $0 \le I \le 10$ then $P_{12}(\x) > -2n_{\N}$.
\end{lemma}
\begin{proof}
    Let $u_i$ be a vertex of $\T$ that is not a leaf, and let $(u_i, u_l),
    (u_i, u_h) \in E(\T)$. Consider the subgraph $G_i$ of $\N$ that is
    induced by $d(\x, u_i)$. It follows from Lemmas~\ref{P4} and~\ref{P6}
    that each connected component of this subgraph is a directed path. Let
    $c_i$ be the number of connected components of $G_i$.  Then, both sums
    $$\sum_{j= 0}^{n_{\N}-1} \sum_{ \substack{ k=0 \\ k \neq j}}^{n_{\N}-1} g(v_j, v_k)
    x_{i,j} x_{l,k} \text{ and } \sum_{j= 0}^{n_{\N}-1} \sum_{ \substack{
    k=0 \\ k \neq j}}^{n_{\N}-1} g(v_j, v_k)  x_{i,j} x_{h,k}$$
    in $P_{12}(\x)$ are each at most $c_i$ because, by Lemmas~\ref{P7}
    and~\ref{P9}, for each connected component there exists at most one
    edge from a vertex of this component to a vertex in $d(\x, u_l)$ and at
    most one edge from a vertex of this component to a vertex in $d(\x, u_h)$.
    Moreover, summing over all non-leaf vertices of $\T$ we have $$\sum_{ \substack {u_i \text{ is not a } \\ \text{ leaf of \(\T\) } }}c_i\leq n_\N$$ \blue{and hence}
    \begin{align*}
        \centering
        & \sum_{i=0}^{n_{\T}-1}  \sum_{\substack{l=0 \\ l \neq i}}^{n_{\T}-1} \Biggl( f(u_i, u_l) \Biggl( - \sum_{j= 0}^{n_{\N}-1} \sum_{ \substack{ k=0 \\ k \neq j}}^{n_{\N}-1} g(v_j, v_k)  x_{i,j} x_{l,k}  \Biggr) \Biggr) \ge -2n_{\N} \\
        \implies P_{12}(\x) =& \sum_{i=0}^{n_{\T}-1}  \sum_{\substack{l=0 \\ l \neq i}}^{n_{\T}-1} \Biggl( f(u_i, u_l) \Biggl( 1 - \sum_{j= 0}^{n_{\N}-1} \sum_{ \substack{ k=0 \\ k \neq j}}^{n_{\N}-1} g(v_j, v_k)  x_{i,j} x_{l,k}  \Biggr) \Biggr) > -2n_{\N}
    \end{align*}\end{proof}

Now, with the last lemma in mind, throughout the remainder of this section,
let \(A = 2n_{\N}\), and let \cyan{$$ B = 4n_{\N}^2 n_{\T}^2 > 2 \cdot \min\{A \cdot
\red{(}- P_{11}(\x)\red{)}, - P_{12}(\x) \} \;\; \allowbreak \text{ for each } \x \in
\mathbb{B}^{m}.$$ }

We next establish three lemmas.

\begin{lemma} \label{lemma: P5 implies H}
    Let \(\x \in \mathbb{B}^m\). If \(H(\x) = 0\) then \(P_I(\x) = 0\) for each \( 1 \leq I \leq 10\).
\end{lemma}
\begin{proof}
    Suppose that \(H(\x)\) = 0. By the definition of $H(\x)$, we have \(P_I(\x) \geq 0\) for all \( 1 \leq I \leq 10\). Now assume that \(P_I(\x) > 0\) for some \(1 \leq I \leq 10\). Then $$H(\x)=B\cdot  \Biggr(\sum_{I= 1}^{10}   P_I(\x)\Biggl) + A \cdot P_{11}(\x) + P_{12}(\x) > 0;$$ a contradiction and the lemma follows.
\end{proof}

\begin{lemma} \label{P11=0}
   Let \(\x \in \mathbb{B}^m\). If \(H(\x) = 0\) then \(P_{11}(\x) = 0\).
\end{lemma}
\begin{proof}
Suppose that \(H(\x)\) = 0. By Lemma \ref{lemma: P5 implies H}, it follows that \(P_I(\x) = 0\) for each \(1 \leq I \leq 10\). First, assume that \(P_{11}(\x) < 0\). There are two cases to consider.
     
\noindent \textbf{Case 1.} There exists a vertex $u_i\in V(\T)$ such that \(|\D(\x, u_i)| = 0\). Then, \(P_1(\x) > 0\); a contradiction.

\noindent \textbf{Case 2.} There  exist a vertex \(u_i \in V(\T)\) and a vertex \(v_j \in V(\N)\) with \(x_{i,j} = 1\) such that $$\sum_{\substack{k=0 \\ k \neq j }}^{n_{\N} -1} g(v_j, v_k)x_{i,k}  \allowbreak \geq 2$$ in $P_{11}(\x).$ Hence, there exist two edges \((v_j, v_a), (v_j, v_b) \in E(\N) \) such that  \(\{ v_j, v_a, v_b\}\subseteq \D(x, u_i)\). As, \(P_3(\x) = 0\), this implies that  \(P_4(\x) > 0\); another contradiction. 

Second, assume that \(P_{11}(\x) > 0\). 
By Lemma \ref{lemma: min value of P12}, $P_{12}(\x) > -2n_{\N}$. Now, as \(P_{11}(\x) > 0\) and \(A\cdot P_{11}(\x) + P_{12}(\x) > 0\), with \(A = 2n_{\N}\), it follows that \(P_{11}(\x) > 0\) implies that \(H(\x) > 0\); a final contradiction. This establishes the lemma.
\end{proof}

\begin{lemma}\label{P12=0}
   Let \(\x \in \mathbb{B}^m\). If \(H(\x) = 0\) then \(P_{12}(\x) = 0\).
\end{lemma}
\begin{proof}
   Suppose that \(H(\x)\) = 0. It follows from Lemmas~\ref{lemma: P5 implies H} and~\ref{P11=0}  that  \(P_I(\x) = 0\) for each \(1 \leq I \leq 11\). Hence, if \(H(\x) = 0\), then \(P_{12}(\x) = 0\).
\end{proof}

For the next theorem, we need a new definition. Let $\T$ be a leaf-labeled rooted tree, and let $X$ be the leaf set of $\T$. For a vertex $u$ of $\T$, we use $C_\T(u)$ to denote the subset of $X$ that precisely contains each element of $X$ that is a descendants of $u$. Note that, if $u$ is a leaf of $\T$ with label $x$, then $C_\T(u)=\{x\}$. Moreover, if $\T$ is a phylogenetic tree, then it immediately follows that $C_\T(u)\ne C_\T(u')$ for two distinct vertices of $\T$. We are now in a position to establish the main result of this section.

\begin{theorem}
Let $\N$ be a phylogenetic network on $X$, and let $\T$ be a phylogenetic $X$-tree. Then, for each $\x \in \mathbb{B}^m$,
    \(H(\x) = 0\) if and only if \(\N\) displays \(\T\).
\end{theorem}

\begin{proof}
(\(\implies\)) Suppose that \(H(\x) = 0\). Then, by Lemmas~\ref{lemma: P5 implies H}--\ref{P12=0}, it follows that \(P_I(\x) = 0\) with \( 1 \leq I \leq 12\). We  start by deleting edges and vertices in $\N$ as described by the following 2-step process.
\begin{enumerate}[(1)]	
\item Delete each vertex in \(\D(\x, u_{n_\T})\) and each edge that is incident with at least one vertex in \(\D(\x, u_{n_\T})\) in $\N$. Let $\N_1$ be the resulting graph.
\item For each ordered pair \((u_a, u_b)\) of vertices in \(V( \T) \) such that \((u_a, u_b) \notin E(\T)\), delete each edge  from a vertex in \(d(\x, u_a)\) to a vertex in  \(\D(\x, u_b)\) in \(\N_1\).  Let $\N_2$ be the resulting graph.
\end{enumerate}
\blue{Recall that,} by Lemma~\ref{P1}, $|d(\x,u_0)|=1$, where $u_0$ is the root of $\T$ and, by Lemma~\ref{P10}, for each leaf $u_i$ of $\T$, $d(\x,u_i)$ contains the leaf of $\N$ that has the same label as $u_i$. We next obtain a graph $\N_3$ from $\N_2$ such that $\N_3$ is a subdivision of $\T$. For each $u_l\in V(\T)\setminus\{u_0\}$, it follows from Lemma~\ref{P11}, that $d(\x,u_l)$ is a directed path of $\N$ and therefore, by construction, also of $\N_2$. Let $u_i$ be the unique ancestor of $u_l$ in $\T$. Let $p=p_1,p_2,\ldots,p_k$ be the directed path of $\N$ that is induced by $d(\x,u_i)$ and, similarly, let $p'=p'_1,p'_2,\ldots,p'_{k'}$ be the directed path of $\N$ that is induced by $d(\x,u_l)$. By Lemma~\ref{P12} and Corollary~\ref{cor:connect}, there exists exactly one edge $e$ in $\N$ that joins a vertex of  $p$ to a vertex of $p'$ and, in particular, $e$ is directed out of $p_k$. Hence $e=(p_k,p'_j)$ for some $1\leq j\leq k'$. As  $u_i$ is the unique parent of $u_l$ in $\T$, any edge in $\N$ that is directed into a vertex in $\{p'_1,p'_2,p'_{j-1},p'_{j+1},\ldots,p'_{k'}\}$ and does not lie on $p'$ has been deleted in Step (2) above. Moreover, again by  Lemma~\ref{P12} and Corollary~\ref{cor:connect}, any edge in $\N$ that joins a vertex in $d(\x,u_l)$ with a vertex in $d(\x,u_{l'})$, where $u_{l'}$ is a child of $u_l$ in $\T$, is directed out of $p'_{k'}$. Hence, any edge in $\N$ that is directed out of a vertex in  $\{p'_1,p'_2,\ldots,p'_{k'-1}\}$ and does not lie on $p'$ has also been deleted in Step (2) above. For the upcoming construction step, we call  $\{p'_1,p'_2,\ldots,p'_{j-1}\}$ the set of {\it dangling} vertices in $\N_2$ with respect to $u_l$. This set may or may not be empty. Now obtain a graph $\N_3$ from $\N_2$ by deleting the set of dangling vertices in $\N_2$ for each vertex $u_l\in V(\T)\backslash\{u_0\}$. It is  straightforward to check that $\N_3$ is a subdivision of $\T$ and that $\T$ can be obtained from this subdivision by suppressing each vertex with in-degree~1 and out-degree~1 in $\N_3$.  Thus, if $H(\x)=0$, then $\N$ displays~$\T$.

(\(\impliedby\)) Suppose that \(\N\) displays \(\T\). Then there exists a subset $V$ of $V(\N)$ and a subset $E$ of $E(\N)$ such that  $\T$ can be obtained from $\N$ by deleting each vertex in $V$ and each edge in $E$ from $\N$ and, subsequently, suppressing any resulting vertex of in-degree 1 and out-degree 1. Without loss of generality, we choose $V$ such that its size $|V|$ is maximized.  Let $\{u_0,u_1,\ldots,u_{n_\T-1}\}$ be the vertices of $\T$, and let $\{v_0,v_1,\ldots,v_{n_\N-1}\}$ be the vertices of $\N$. Furthermore, let $\T'$ be the graph obtained from $\N$ by deleting each vertex in $V$ and each edge in $E$ from $\N$. By construction, $\T'$ is a subdivision of $\T$. Consider the map $m: V(\T) \rightarrow 2^{V(\T')}$ that maps each $u_i\in V(\T)$ to a subset of $V(\T')$ such that  $m(u_i)$ contains precisely each vertex $v_j$ of $\T'$ with $C_{\T'}(v_j)=C_\T(u_i)$. Now we define $\blue{\x\in\mathbb{B}^m}$. For each  $x_{i,j}$ with $0\leq i\leq n_\T-1$ and $0\leq j\leq n_\N-1$, we set $x_{i,j}=1$ if $v_j\in m(u_i)$ and, $x_{i,j}=0$ otherwise. Moreover, for each  $0\leq j\leq n_\N-1$, we set $x_{n_\T,j}=1$ if $v_j\in V$ and $x_{n_\T,j}=0$ otherwise. Lastly, to keep with the notation introduced in Section~\ref{sec: QUBO Formulation}, let $d(\x,u_i)=m(u_i)$ for each $i$ with $0\leq i\leq n_\T-1$, and let $d(\x,u_{n_\T})=V$.
    
We complete the proof by showing that \(P_I(\x) = 0\) for each \( 1 \leq I \leq 12\).
    \begin{enumerate}[(1)]
        \item \blue{As $\T'$ is a subdivision of $\T$,  the root of $\T'$ is the only vertex of $\T'$ whose set of leaf descendants is $X$ and, so $|d(\x,u_0)|=1$. Furthermore, by construction,} $|d(\x,u_i)|>0$ for each $u_i\in V(\T)$. By Lemma~\ref{P1}, \(P_1(\x) = 0\) follows.
        \item For each vertex $v_j\in V(\T')$ there exists exactly one vertex  $u_i\in V(\T)$ such that $C_{\T'}(v_j)=C_\T(u_i)$. Moreover $d(\x,u_{n_\T})=V$. Hence \(P_2(\x) = 0\) follows from Lemma~\ref{P2}.
        \item  Set \(z_{i,j} = x_{i, \childIndex{1}{j}} x_{i, \childIndex{2}{j}} \), where \(u_i \in V(\T)\) and \(v_j\in V(\N)\) is tree vertex. By Lemma~\ref{P3}, this implies that \(P_3(\x) = 0\). Similarly, set \(z_{i,j} = x_{i, \parentIndex{1}{j}} x_{i, \parentIndex{2}{j}}\), where \(u_i \in V(\T)\) and \(v_j\in V(\N)\) is reticulation vertex. Then, by Lemma~\ref{P5}, we have \(P_5(\x) = 0\). Lastly, set \(\hat{z}_{i,2j} = x_{i,j} x_{i, \childIndex{1}{j}}\) and \(\hat{z}_{i,2j+1} = x_{i,j} x_{i, \childIndex{2}{j}}\), where \(u_i \in V(\T)\) and \(v_j\in V(\N)\) is a tree vertex. It then follows by Lemma~\ref{P8} that \(P_8(\x) = 0\).
        \item Since $\T'$ is a subdivision of $\T$, it follows from the definition of \blue{$m$ (and consequently $d$)} that, for each $0\leq i\leq n_\T-1$, the subgraph of $\N$ that is induced by the vertices in $d(\x,u_i)$ is a directed path. Now, towards a contradiction, assume that there exist \blue{a vertex $u_i$ in $\T$ and a vertex} $v_j$ in $\N$ such that $\{\{v_j,c_1(v_j),c_2(v_j)\}\subseteq d(\x,u_i)$. Since the subgraph of $\N$ that is induced by the vertices in $d(\x,u_i)$ is a directed path, we may assume without loss of generality that there is a directed path $(c_1(v_j)=p_1,p_2,\ldots, p_k= c_2(v_j))$ in $\N$ with $k\geq 2$. In particular, each vertex on this path is contained in $d(\x,u_i)$. It now follows that we can  obtain a subdivision of $\T$ from $\N$ by deleting all vertices in $V\cup \{p_1,p_2,\ldots,p_{k-1}\}$ and edges in $$(E\backslash (v_j,p_k))\cup \{(v_j,p_1),(p_1,p_2),\ldots,(p_{k-1},p_k)\}.$$ This contradicts the maximality of $|V|$. Hence, there exist no two vertices $u_i$ and $v_j$ such that $\{v_j,c_1(v_j),c_2(v_j)\}\subseteq d(\x,u_i)$. An analogous contradiction can be used to establish that there  exist no two vertices $u_i$ in $\T$ and $v_j$ in $\N$ such that $\{v_j,p_1(v_j),p_2(v_j)\}\subseteq d(\x,u_i)$  Thus, by Lemmas~\ref{P4}, \ref{P6}, and \ref{P11}, we have \(P_4(\x) = P_6(\x) = P_{11}(\x) = 0\). 
        \item Again, since $\T^{\prime}$ is a subdivision of $\T$, for each $(u_a, u_b) \in E(\T)$, there exists an edge that joins the terminal vertex $v_r$ of the directed path in $\N$ induced by \(d(\x,u_a)\) to the first vertex $v_s$ of the directed path in $\N$ induced by \(d(\x,u_b)\). We now move towards several contradictions. Assume that there exist more than one edge from a vertex of the directed path in $\N$ induced by \(d(\x,u_a)\) to a vertex of the directed path in $\N$ induced by \(d(\x,u_b)\). This is, there exist $v_q, v_r \in d(\x, u_a)$ and $v_s, v_t \in d(\x, u_b)$ such that $(v_q, v_t), (v_r, v_s) \in E(\N)$. Let $v_q=p_1,p_2,\ldots,p_k=v_r$ be the directed path from $v_q$ to $v_r$ in $\N$ and, similarly, let $v_s=p'_1,p'_2,\ldots,p'_{k'}=v_t$ be the directed path from $v_s$ to $v_t$ in $\N$. Since $\N$ has no edges in parallel, $k\geq 2$ or $k'\geq 2$. If $v_q=v_r$, then $v_s\ne v_t$ and, consequently, a subdivision of $\T$ can be obtained from $\N$ by deleting all vertices in $V\cup \{p'_1,p'_2,\ldots,p'_{k'-1}\}$ and all edges in $E\cup \{(p'_1,p'_2),(p'_2,p'_3),\ldots,(p'_{k'-1},p'_{k'})\}$; thereby contradicting the maximality of $|V|$. Similarly, if $ v_s=v_t$, then we obtain  a subdivision of $\T$ from $\N$ by deleting all vertices in $V\cup \{p_2,p_3,\ldots,p_{k}\}$ and all edges in $E\cup \{(p_1,p_2),(p_2,p_3),\ldots,(p_{k-1},p_{k})\}$; again contradicting the maximality of $|V|$. We may therefore assume without loss of generality that $k \ge 2$ and $k'\geq 2$.  Now recall that $(v_q, v_t)$ and $(v_r, v_s)$ are edges in $\N$. Then a subdivision of $\T$ can be obtained from $\N$ by deleting all vertices in $V\cup \{p_2,p_3,\ldots,p_{k},p'_1,p'_2,\ldots,p'_{k'-1}\}$ and all edges in $E\cup \{(p_1,p_2),(p_2,p_3),\ldots,(p_{k-1},p_{k}),(p'_1,p'_2),(p'_2,p'_3),\ldots,(p'_{k'-1},p'_{k'})\}$; a final contradictions.
        It follows that there exists exactly one edge from a vertex in \(d(\x,u_a)\) to a vertex in \(d(\x,u_b)\). In particular, this edge joins the terminal vertex of the directed path in $\N$ induced by \(d(\x,u_a)\) to the first vertex of the directed path in $\N$ induced by \(d(\x,u_b)\). Hence, by Lemmas~\ref{P7}, \ref{P9}, and \ref{P12}, we have \(P_7(\x) = P_9(\x) = P_{12}(\x)= 0\).

        \item Clearly, for each $u_i\in V(\T)$ that is a leaf, we have $v_j\in d(\x,u_i)$, where $v_j$ is the leaf vertex of $\T'$ (and $\N$) whose label is identical to that of $u_i$. Hence, by Lemma~\ref{P10}, it follows that  \(P_{10}(\x)=0\). 
    \end{enumerate}
\noindent Therefore, if \(\N\) displays \(\T\), then \(H(\x) =0\).
\end{proof}

We end this section by noting that $H(\x)$ contains constant and \cyan{quadratic terms}. Strictly speaking, $H(\x)$ is therefore not in QUBO form. \cyan{However, $H(\x)$ can easily be converted into QUBO form by deleting all the constant terms. The summation of all the constant terms that we deleted is called the offset; we use it during post-processing.}

\subsection{\cyan{Solving our QUBO Formulation using a Quantum Annealer}}

\cyan{A quantum annealer finds the minimum energy state of a QUBO. We saw in the last section that we can convert $H(\x)$ to an objective funciton in QUBO form. To show that we can solve Tree Containment using a quantum annealer, it suffices to show that $H(\x) \geq 0$ for all $\x \in \mathbb{B}^m$.}

\begin{proposition}
    \cyan{$H(\x) \geq 0$ for all $\x \in \mathbb{B}^m$.}
\end{proposition}
\begin{proof}
    \cyan{Let $\x \in \mathbb{B}^m$. By definition, $P_I(\x) \geq 0$ for all $I \in \{1,2,\ldots,10\}$. If $P_I(\x) > 0$ for any $I \in \{1,2,\ldots,10\}$ then $H(\x) >0$, so we assume that $P_I(\x) = 0$ for all $I \in \{1,2,\ldots,,10\}$. Now, it suffices to show that either $P_{11}(\x) > 0$, or $P_{11}(\x) = 0$ and $P_{12}(\x) \geq 0$.}

    \cyan{Suppose $P_{11}(\x) < 0$. Then, for some  $x_{i,j} = 1$, we have $$\sum_{\substack{k=0 \\ k \neq j }}^{n_{\N}-1} g(v_j, v_k) x_{i,k} \geq 2 $$in $P_{11}(\x)$. This implies that $P_3(\x)>0$ or $P_4(\x)>0$, which gives us a contradiction, so $P_{11}(\x)\geq0$. If $P_{11}(\x) > 0$ and $P_I(\x) = 0$ for all $I \in \{1,2,\ldots,10\}$ then $H(\x) > 0$, so we assume that $P_{11}(\x) = 0$.}

    \cyan{Suppose $P_{12}(\x) < 0$. Then, for some edge $(u_i, u_l) \in E(\T)$, we have 
    $$\sum_{j= 0}^{n_{\N}-1} \sum_{ \substack{ k=0 \\ k \neq j}}^{n_{\N}-1} g(v_j, v_k) x_{i,j} x_{l,k} \geq 2$$
    in $P_{12}(\x)$. Since $P_I(\x) = 0$ for all $I \in \{1,2,\ldots,11\}$, it follows by Lemma \ref{P12} that 
    $$\sum_{j= 0}^{n_{\N}-1} \sum_{ \substack{ k=0 \\ k \neq j}}^{n_{\N}-1} g(v_j, v_k) x_{i,j} x_{l,k} \leq 1, $$ which gives us a contradiction, so $P_{12}(\x) \geq 0$.}
\end{proof}

\section{Example and Results} \label{sec: results}
In this section, we present an example to illustrate the QUBO formulation of Tree Containment, \cyan{some minor embedding results, and the post-processing procedure}.
\subsection{Example}

Consider the phylogenetic network \(\N\) on \(X\) and the two phylogenetic \(X\)-trees \(\T_1\) and \(\T_2\) as shown in Figure \ref{fig: example 2 with vertex indecies}. We wish to answer the following two questions: (i) Does \(\N\) display \(\T_1\)? (ii) Does \(\N\) display \(\T_2\)?

\begin{figure}[h]
    \centering
    \begin{subfigure}[b]{0.3\textwidth}
        \centering
        \includegraphics[width=\textwidth]{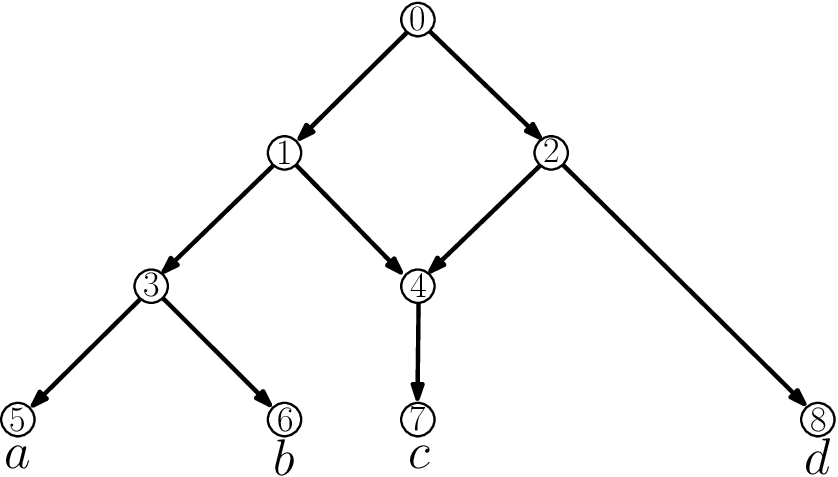}
        \caption{\(\N\)}
        \label{fig: example (a)}
    \end{subfigure}
    \hfill
    \begin{subfigure}[b]{0.3\textwidth}
        \centering
        \includegraphics[width=\textwidth]{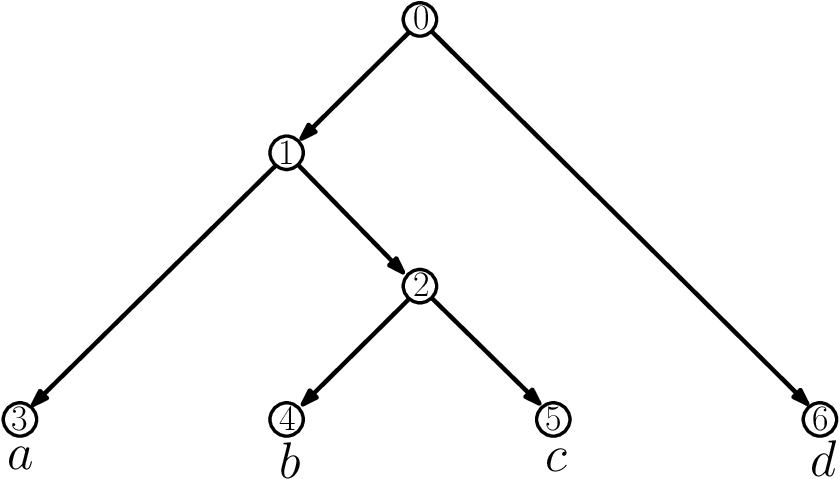}
        \caption{\(\T_1\)}
        \label{fig: example (b)}
    \end{subfigure}
    \hfill
    \begin{subfigure}[b]{0.3\textwidth}
        \centering
        \includegraphics[width=\textwidth]{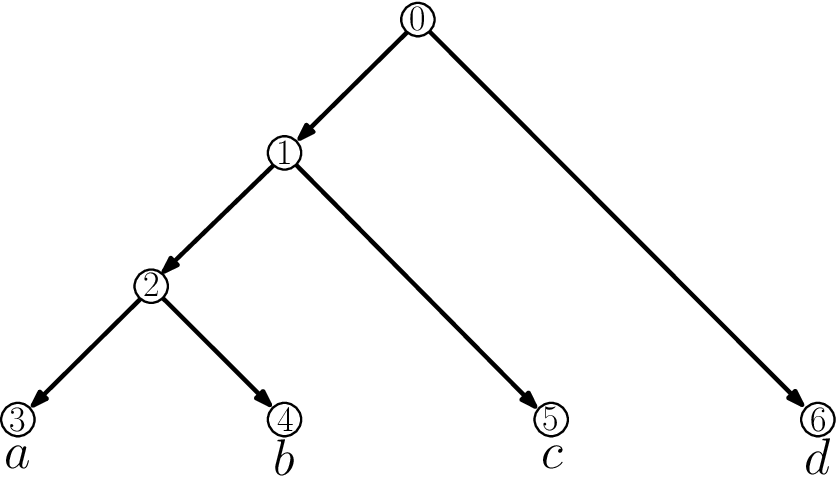}
        \caption{\(\T_2\)}
        \label{fig: example (c)}
    \end{subfigure}
       \caption{A phylogenetic network $\N$ on $X$ in (a) and two phylogenetic $X$-trees $\T_1$ and $\T_2$ in (b) and (c), where \(X = \{a, b, c, d\}\). Each vertex in $\N$, $\T_1$, and $\T_2$ has been assigned a number. The sole purpose of those numbers, which are not vertex labels, is the representation of \(\N\), \(\T_1\), and \(\T_2\) as adjacency matrices.} 
       \label{fig: example 2 with vertex indecies}
\end{figure}

After processing \blue{of $H(\x)$ as described in the last paragraph of Section~\ref{sec:main-result}, the resulting objective function is in} QUBO form. Recall that $$\x \in \mathbb{B}^{n_{\N}(n_{\T}+1) + (n_{\T} - 1 )(1 + \left\lfloor \lg(n_{\N}-n_{\T})\right\rfloor) + n_\T( \alpha + \beta)+ 2 \beta \gamma}.$$ Then, for each of (i) and (ii), we need $$9\cdot(7+1)+(7-1)(1+\lg(2))+7\cdot(1+4)+2\cdot 4\cdot 3=143$$ logical qubits and, so, the size of the output QUBO matrix is \(143 \times 143\). 

For Question (i), D-Wave's quantum annealer returned only one binary vector  \(\x^{*}\)  as follows, \blue{where subscripts refer to the vertex numbers as shown in Figure~\ref{fig: example 2 with vertex indecies}}: \(x_{3,5} = x_{2,4} = x_{1,1} = x_{3,9} = x_{6,8} = x_{4,6} = x_{5,7} = x_{0,0} = x_{6,2} = x_{6,9} = x_{3,3} = 1\) and all  other qubits are 0. After post-processing, we did not get \(0\) as minimum output  \(x^{*}\). We confirmed this result by querying D-Wave's quantum annealer 100 times and, each time, the result was \(x^{*}=1\). Therefore, there is a high probability that $\N$ does not display $\T_1$. Indeed, \(\N\) does not display \(\T_1\).

For Question (ii), D-Wave's quantum annealer returned a binary vector  \(\x^{*}\) with \( x_{3,5} = x_{1,1} = x_{6,8} = x_{4,6} = x_{5,7} = x_{0,0} = x_{5,4} = x_{5,9} = x_{2,3} = x_{6,2} = x_{6,9} = 1\) and all  other qubits are 0. After post-processing, we did get \(0\) as minimum output \(x^{*}\). It follows that \(\N\) displays \(\T_2\) which is indeed the case. 

\subsection{Minor Embedding Results}
As a proof-of-concept, we performed several experiments to investigate the number of logical and physical qubits that are needed to solve an instance of Tree Containment. We start by providing some information about the host graphs. After formulating  an instance of Tree Containment as an instance of QUBO, we used the minor embedding algorithm provided by D-Wave \cite{minorminer}. For some of our test cases (described below), the QUBO matrix could not be embedded into the host graph of D-Wave Advantage. We therefore considered a larger Pegasus graph. This allows us to compare the current capacity of the D-Wave Advantage annealer with a possible future version of the D-Wave machine. Both host graphs that we considered were of Pegasus topology\footnote{\cyan{After this paper was written, a new topology was released by D-Wave called the Zephyr topology, which is an improvement over the Pegasus topology.}}. The first host graph, D-Wave Advantage, has 5640 vertices and 40484 edges, and the second host graph has 23560 vertices and 172964 edges. The first and second host graphs are represented by \(P16\) and \(P32\), respectively. In comparison, \(P32\) has about 4.17 times more vertices and about 4.27 times more edges than \(P16\). For more information about \blue{the} Pegasus topology, see \cite{D-Wave_graphs}.

In total, we have analyzed 15 small instances of Tree Containment, where each instance consists of a phylogenetic network $\N$ on $X$ and a phylogenetic $X$-tree $\T$.  In Table \ref{tbl: density and logical qubits}, the first column contains the size of \(X\), the second column contains the number $r$ of reticulation vertices in \(\N\), \cyan{the \red{third} column contains the number $s$ of additional logical qubits required due to the conversion from cubic to quadratic (it is equal to $n_\T(\alpha + \beta) + 2\beta\gamma$)}, the \red{fourth} column contains the \red{total} number of binary variables in the QUBO instance obtained from applying the approach described in Section \ref{sec: QUBO Formulation} to $\N$ and $\T$, \cyan{and the fifth column contains the number of couplers (non-zero off diagonal entries) in the QUBO matrix.} \blue{Finally, the last column contains the density of the graph $G$ whose weighted adjacency matrix is the QUBO matrix, where the density is defined as the ratio of the size of $G$ and the size of the complete graph whose order is the same as that of $G$.}
Note that, in the approach described in Section \ref{sec: QUBO Formulation}, the number of binary variables in the resulting QUBO instance  depends only \blue{linearly} on the total number of vertices in \(\N\) and \(\T\). As $\N$ and $\T$ both have leaf set $X$, it follows that  $n_\T=2|X|-1$ and $n_\N=2|X|+2r-1$ vertices~\cite[Lemma 2.1]{mcdiarmid2015counting}. 
From our initial experiments, we can see some good news in that the number of logical qubits are  \blue{expectantly small and that the densities are relatively low}.  This later fact \blue{implies that} fewer qubit connections will be required and non-completely connected quantum annealers will have a better chance of embedding the logical structure onto a physical structure.
In Table \ref{tbl: minor embedding results}, we present the minor embedding results for both host graphs \(P16\) and \(P32\). \cyan{For each instance, the minor embedding algorithm \cite{minorminer} was run $10$ times with the \textit{timeout} parameter as $240$ (seconds). We present the best out of the $10$ runs in terms of the physical qubits required by the instance.} The first two columns of this table are identical to the first two columns of Table~\ref{tbl: density and logical qubits}. The  \textit{Physical Qubits} column contains our experimental results indicating the number of physical qubits required to embed a QUBO instance, depending on which host graph was used. The \textit{Max Chain Size} column contains the maximum number of physical qubits a single logical qubit was mapped onto in our experiments. The entries marked by `-' correspond to the cases where the minor embedding algorithm was not able to find a minor embedding. \cyan{Lastly, the \textit{Average Time} column contains the average time taken to get the minor embedding. Note that some of the entries go beyond the time-limit of $240$ (seconds), we think that this is due to the minor embedding algorithm \cite{minorminer} using the \textit{timeout} parameter as a soft bound.}
Regrettably, but not unexpectedly, the number of physical qubits grows beyond the capabilities of current quantum annealing architectures such as those used by D-Wave, even for  small test cases.  However, \blue{it is worth noting that} we do not have exponential growth with respect to the input sizes. 
Furthermore, the relatively large maximum chain sizes is of a concern, but with advances in the expected qubit interconnection density of future hardware (and possibly improved embedding algorithms) this can \blue{likely} be mitigated in practice. 

\begin{table}
    \caption{The number of logical qubits and QUBO density.}
    \label{tbl: density and logical qubits}
    \begin{tabularx}{\textwidth}{  >{\centering\arraybackslash}X
        >{\centering\arraybackslash}X
        >{\centering\arraybackslash}X
        >{\centering\arraybackslash}X
        >{\centering\arraybackslash}X
        >{\centering\arraybackslash}X}
        \toprule
        $|X|$ & $r$ & $s$ & Logical Qubits & Couplers & Density\\
        \midrule
        4  & 1 & 59  & 143  & 897   &  0.088348 \\ 
        4  & 3 & 99  & 221  & 1621  &  0.066680 \\
        4  & 2 & 79  & 185  & 1278  &  0.075088 \\
        5  & 3 & 146 & 320  & 2690  &  0.052704 \\
        5  & 4 & 172 & 374  & 3385  &  0.048530 \\
        5  & 2 & 120 & 274  & 2191  &  0.058581 \\
        6  & 1 & 137 & 313  & 2647  &  0.054211 \\
        6  & 5 & 265 & 557  & 5841  &  0.037721 \\
        6  & 3 & 201 & 435  & 4111  &  0.043551 \\
        7  & 2 & 226 & 500  & 5049  &  0.040473 \\
        7  & 4 & 302 & 644  & 7155  &  0.034558 \\
        7  & 3 & 264 & 566  & 5932  &  0.037099 \\
        8  & 5 & 423 & 879  & 10995 &  0.028493 \\
        8  & 4 & 379 & 803  & 9736  &  0.030236 \\
        8  & 3 & 335 & 713  & 8201  &  0.032309 \\
        \bottomrule
    \end{tabularx}    
\end{table}

\begin{table}
    \caption{Minor embedding results.}
    \label{tbl: minor embedding results}
    \begin{tabularx}{\textwidth}{  >{\centering\arraybackslash}X
        >{\centering\arraybackslash}X
        >{\centering\arraybackslash}X
        >{\centering\arraybackslash}X
        >{\centering\arraybackslash}X
        >{\centering\arraybackslash}X
        >{\centering\arraybackslash}X
        >{\centering\arraybackslash}X}
      \toprule
      \multicolumn{1}{c}{$|X|$}                 &
      \multicolumn{1}{c}{$r$}                   &
      \multicolumn{2}{c}{Physical Qubits }      &
      \multicolumn{2}{c}{Max Chain Size }       &
      \multicolumn{2}{c}{Average Time (seconds) }  \\
      \cline{3-4}
      \cline{5-6}
      \cline{7-8}
         &   & {\(P16\)} & {\(P32\)} & {\(P16\)} & {\(P32\)} & {\(P16\)} & {\(P32\)} \\
      \midrule
      4  & 1 & 639   & 637      & 12       & 12     & 32.3   & 51.3       \\
      4  & 3 & 1094  & 1221     & 13       & 17     & 69.5   & 173.5      \\
      4  & 2 & 910   & 888      & 15       & 14     & 48.5   & 139.4      \\
      5  & 3 & 2529  & 2349     & 27       & 23     & 100.8  & 243.2      \\
      5  & 4 & 2992  & 3333     & 27       & 30     & 94.5   & 243.8      \\
      5  & 2 & 1867  & 1841     & 20       & 19     & 92.0   & 235.9      \\
      6  & 1 & 2651  & 2884     & 28       & 29     & 92.3   & 242.9      \\
      6  & 5 & {-}   & 8796     & {-}      & 65     & 246.9  & 251.2      \\
      6  & 3 & 4428  & 5098     & 39       & 40     & 145.9  & 252.1      \\
      7  & 2 & {-}   & 7073     & {-}      & 53     & 245.8  & 246.7      \\
      7  & 4 & {-}   & 11153    & {-}      & 73     & 248.1  & 247.6      \\
      7  & 3 & {-}   & 8968     & {-}      & 55     & 245.9  & 254.5      \\
      8  & 5 & {-}   & {-}      & {-}      & {-}    & 253.8  & 322.8      \\
      8  & 4 & {-}   & {-}      & {-}      & {-}    & 248.4  & 307.6      \\
      8  & 3 & {-}   & 13276    & {-}      & 55     & 249.1  & 290.1      \\
      \bottomrule
    \end{tabularx}    
\end{table}

\subsection{Post-Processing}
\cyan{Suppose that we have a phylogenetic network $\N$ on $X$ and a phylogenetic $X$-tree $\T$. Using $\N$ and $\T$, we construct $H(\x)$ as described in Section \ref{sec: QUBO Formulation}. We then process $H(\x)$ as described in the last paragraph of Section \ref{sec:main-result}. This gives us an objective function in QUBO form and an offset. In post-processing, we simply add the offset to the minimum value of the objective function. If the offset plus the minimum value of the objective function equals 0, then $\N$ displays $\T$ and, otherwise, $\N$ does not display $\T$.}

\cyan{Suppose $\N$ displays $\T$. We can get an explicit mapping from the vertices of $\T$ to the vertices of $\N$ which shows that $\N$ displays $\T$ using the formulation described in Section \ref{sec: QUBO Formulation}. Let $\x$ be a input value for which the objective function attains its minimum. We use the map $d: (\mathbb{B}^m,
\allowbreak V(\T) \cup \{u_{n_\T}\})\rightarrow 2^{V(\N)}$ as defined in Section \ref{sec: QUBO Formulation} to map the vertices of $\T$ to the vertices of $\N$. For example, the root vertex of $\T$, that is $u_0$, would be mapped to $d(\x, u_0)$, which is a directed path in $\N$. Similarly, we find the mapping for each other vertex in $\T$. We delete all the vertices of $\N$ that are not mapped by any vertex of $\T$ and the edges incident to them. Then, our formulation ensures that we can suppress \red{any vertex with in-degree 1 and out-degree 1 of} the directed path $d(\x, u)$ for any $u \in V(\T)$. 
This would explicitly show that \red{we can derive $\T$ from $\N$ by deleting edges and vertices and suppressing any resulting vertices of in-degree 1 and out-degree 1.} Hence, $\N$ displays $\T$.}

\section{Conclusion} \label{sec: conclusion}
In this paper, we have discussed the AQC model, which has gained popularity in recent years, to solve Tree Containment. Currently, the size of  problem instances (of Tree Containment as well as other problems) that can be solved using quantum annealers is one of the main setbacks that is hindering quantum computing from becoming a mainstream technology. The development of different approaches to build bigger and more efficient quantum annealers is an ongoing and major effort. 

In Section~\ref{sec: TC problem}, we have established an efficient reduction from Tree Containment to QUBO. 
A Python program that reads in a phylogenetic network and tree and outputs the resulting QUBO instance is given in  Appendix \ref{appendix}. 
We have shown that an instance of Tree Containment, say $(\N,\T)$, can be reduced to an instance of QUBO whose number of logical qubits  is $O(n_\N n_\T)$
Furthermore, the reduction from Tree Containment to QUBO takes polynomial time.   
Our reduction has the following special property: solving (minimizing) the QUBO returns 0 after post-processing if and only if \(\N\) displays \(\T\).  In addition, we note that the above QUBO formulation is also correct if \(\T\) has leaf set \(X'\) and \(\N\) has leaf set \(X\) such that \(X' \subset X\). 

In Section~\ref{sec: results}, we have experimentally evaluated the
efficiency of the QUBO formulation for Tree Containment. We have
compared the efficiency in terms of logical qubits, physical qubits,
density, and the maximum chain size. 
Our experiments indicate that current quantum annealers can only solve small
instances of Tree Containment, which we note can also be solved with a classical
approach. Nevertheless, our results provide a first indication that
other problems arising in studying phylogenetic trees and networks could
potentially also be attacked within AQC. For example, many problems that
arise in the reconstruction of phylogenetic networks have  underlying
NP-hard optimization problems. Current algorithms in this area either do
not scale up to large data sets or are heuristics with no guarantee on the
optimality of the solution~\cite{hejase2016scalability}. Hence, AQC offers
a promising and alternative approach to solving such problems. Furthermore,
it would be interesting to explore quantum-classical hybrid approaches,
as discussed in \cite{abbott2019hybrid}. An immediate next step could be
the development of a QUBO for problems that are closely related to Tree
Containment such as the problem of deciding if a given phylogenetic network
contains a given phylogenetic tree as a so-called base tree, or the problem
of deciding if two phylogenetic networks display the same set of phylogenetic
trees~\cite{anaya2016determining,docker2019displaying,francis2015phylogenetic}.
How different are QUBO formulations for these problems from the QUBO
presented in this paper?\\

\noindent{\bf Acknowledgements.} We thank Cristian Calude and Richard Hua for helpful discussions, and two anonymous referees for their constructive comments. The third author was supported by the New Zealand Marsden Fund.

\bibliographystyle{abbrv}

\begin{thebibliography}{10}

\bibitem{abbott2019hybrid}
A.~A. Abbott, C.~S. Calude, M.~J. Dinneen, and R.~Hua.
\newblock A hybrid quantum-classical paradigm to mitigate embedding costs in
  quantum annealing.
\newblock {\em International Journal of Quantum Information}, 17(05):1950042,
  2019.

\bibitem{aharonov2008adiabatic}
D.~Aharonov, W.~Van~Dam, J.~Kempe, Z.~Landau, S.~Lloyd, and O.~Regev.
\newblock Adiabatic quantum computation is equivalent to standard quantum
  computation.
\newblock {\em SIAM Review}, 50(4):755--787, 2008.

\bibitem{anaya2016determining}
M.~Anaya, O.~Anipchenko-Ulaj, A.~Ashfaq, J.~Chiu, M.~Kaiser, M.~S. Ohsawa,
  M.~Owen, E.~Pavlechko, K.~S. John, S.~Suleria, et~al.
\newblock On determining if tree-based networks contain fixed trees.
\newblock {\em Bulletin of Mathematical Biology}, 78(5):961--969, 2016.

\bibitem{bapteste2013networks}
E.~Bapteste, L.~van Iersel, A.~Janke, S.~Kelchner, S.~Kelk, J.~O. McInerney,
  D.~A. Morrison, L.~Nakhleh, M.~Steel, L.~Stougie, et~al.
\newblock Networks: expanding evolutionary thinking.
\newblock {\em Trends in Genetics}, 29(8):439--441, 2013.

\bibitem{blais2021past}
C.~Blais and J.~M. Archibald.
\newblock The past, present and future of the tree of life.
\newblock {\em Current Biology}, 31(7):R314--R321, 2021.

\bibitem{bordewich2016reticulation}
M.~Bordewich and C.~Semple.
\newblock Reticulation-visible networks.
\newblock {\em Advances in Applied Mathematics}, 78:114--141, 2016.

\bibitem{calude2017solving}
C.~S. Calude and M.~J. Dinneen.
\newblock Solving the broadcast time problem using a \mbox{D}-wave quantum
  computer.
\newblock In A.~Adamatzky, editor, {\em Advances in Unconventional Computing.
  Emergence, Complexity and Computation}, volume~22, pages 439--453. Springer,
  2017.

\bibitem{calude2017qubo}
C.~S. Calude, M.~J. Dinneen, and R.~Hua.
\newblock {QUBO} formulations for the graph isomorphism problem and related
  problems.
\newblock {\em Theoretical Computer Science}, 701:54--69, 2017.

\bibitem{D-Wave_graphs}
D-Wave.
\newblock D-{W}ave {QPU} architecture: Topologies.
\newblock \url{https://docs.dwavesys.com/docs/latest/c_gs_4.html}.

\bibitem{minorminer}
D-Wave.
\newblock Minorminer.
\newblock
  \url{https://docs.ocean.dwavesys.com/en/stable/docs_minorminer/source/sdk_index.html}.

\bibitem{D-Wave_qbsolv}
D-Wave.
\newblock qbsolv.
\newblock \url{https://docs.ocean.dwavesys.com/projects/qbsolv/en/latest}.

\bibitem{deutsch1985quantum}
D.~Deutsch.
\newblock Quantum theory, the {C}hurch--{T}uring principle and the universal
  quantum computer.
\newblock {\em Proceedings of the Royal Society of London A. Mathematical and
  Physical Sciences}, 400(1818):97--117, 1985.

\bibitem{docker2019displaying}
J.~D{\"o}cker, S.~Linz, and C.~Semple.
\newblock Displaying trees across two phylogenetic networks.
\newblock {\em Theoretical Computer Science}, 796:129--146, 2019.

\bibitem{farhi2000quantum}
E.~Farhi, J.~Goldstone, S.~Gutmann, and M.~Sipser.
\newblock Quantum computation by adiabatic evolution.
\newblock {\em arXiv preprint quant-ph/0001106}, 2000.

\bibitem{fedorov2021towards}
A.~Fedorov and M.~Gelfand.
\newblock Towards practical applications in quantum computational biology.
\newblock {\em Nature Computational Science}, 1(2):114--119, 2021.

\bibitem{felsenstein2004inferring}
J.~Felsenstein.
\newblock {\em Inferring phylogenies}.
\newblock Sinauer Associates Sunderland, 2004.

\bibitem{feynman2018simulating}
R.~P. Feynman.
\newblock Simulating physics with computers.
\newblock {\em International Journal of Theoretical Physics}, 21:467--488,
  1982.

\bibitem{francis2015phylogenetic}
A.~R. Francis and M.~Steel.
\newblock Which phylogenetic networks are merely trees with additional arcs?
\newblock {\em Systematic Biology}, 64(5):768--777, 2015.

\bibitem{glover2018tutorial}
F.~Glover, G.~Kochenberger, and Y.~Du.
\newblock A tutorial on formulating and using {QUBO} models.
\newblock {\em arXiv preprint arXiv:1811.11538}, 2019.

\bibitem{grover1996fast}
L.~K. Grover.
\newblock A fast quantum mechanical algorithm for database search.
\newblock In {\em Proceedings of the 28th Annual ACM Symposium on Theory of
  Computing}, pages 212--219, 1996.

\bibitem{gunawan2015locating}
A.~D. Gunawan, B.~DasGupta, and L.~Zhang.
\newblock Locating a tree in a reticulation-visible network in cubic time.
\newblock In {\em Proceedings of the 20th Annual Conference on Research in
  Computational Molecular Biology}, page 266, 2016.

\bibitem{hejase2016scalability}
H.~A. Hejase and K.~J. Liu.
\newblock A scalability study of phylogenetic network inference methods using
  empirical datasets and simulations involving a single reticulation.
\newblock {\em BMC Bioinformatics}, 17(1):1--12, 2016.

\bibitem{huson2010phylogenetic}
D.~H. Huson, R.~Rupp, and C.~Scornavacca.
\newblock {\em Phylogenetic networks: concepts, algorithms and applications}.
\newblock Cambridge University Press, 2010.

\bibitem{jetz2012global}
W.~Jetz, G.~H. Thomas, J.~B. Joy, K.~Hartmann, and A.~O. Mooers.
\newblock The global diversity of birds in space and time.
\newblock {\em Nature}, 491(7424):444--448, 2012.

\bibitem{kanj2008seeing}
I.~A. Kanj, L.~Nakhleh, C.~Than, and G.~Xia.
\newblock Seeing the trees and their branches in the network is hard.
\newblock {\em Theoretical Computer Science}, 401:153--164, 2008.

\bibitem{koonin2001horizontal}
E.~V. Koonin, K.~S. Makarova, and L.~Aravind.
\newblock Horizontal gene transfer in prokaryotes: quantification and
  classification.
\newblock {\em Annual Reviews in Microbiology}, 55(1):709--742, 2001.

\bibitem{LucasNP}
A.~Lucas.
\newblock Ising formulations of many \mbox{NP} problems.
\newblock {\em Frontiers in Physics}, 2:5, 2014.

\bibitem{mahasinghe2019solving}
A.~Mahasinghe, R.~Hua, M.~J. Dinneen, and R.~Goyal.
\newblock Solving the hamiltonian cycle problem using a quantum computer.
\newblock In {\em Proceedings of the Australasian Computer Science Week
  Multiconference}, pages 1--9, 2019.

\bibitem{mcdiarmid2015counting}
C.~McDiarmid, C.~Semple, and D.~Welsh.
\newblock Counting phylogenetic networks.
\newblock {\em Annals of Combinatorics}, 19(1):205--224, 2015.

\bibitem{mcgeoch2014adiabatic}
C.~C. McGeoch.
\newblock Adiabatic quantum computation and quantum annealing: Theory and
  practice.
\newblock {\em Synthesis Lectures on Quantum Computing}, 5(2):1--93, 2014.

\bibitem{ottenburghs2019multispecies}
J.~Ottenburghs.
\newblock Multispecies hybridization in birds.
\newblock {\em Avian Research}, 10(1):1--11, 2019.

\bibitem{outeiral2021prospects}
C.~Outeiral, M.~Strahm, J.~Shi, G.~M. Morris, S.~C. Benjamin, and C.~M. Deane.
\newblock The prospects of quantum computing in computational molecular
  biology.
\newblock {\em Wiley Interdisciplinary Reviews: Computational Molecular
  Science}, 11(1):e1481, 2021.

\bibitem{pardalos1992complexity}
P.~M. Pardalos and S.~Jha.
\newblock Complexity of uniqueness and local search in quadratic 0--1
  programming.
\newblock {\em Operations Research Letters}, 11(2):119--123, 1992.

\bibitem{richardson2007horizontal}
A.~O. Richardson and J.~D. Palmer.
\newblock Horizontal gene transfer in plants.
\newblock {\em Journal of Experimental Botany}, 58(1):1--9, 2007.

\bibitem{shor1994algorithms}
P.~W. Shor.
\newblock Algorithms for quantum computation: discrete logarithms and
  factoring.
\newblock In {\em Proceedings of the 35th Annual Symposium on Foundations of
  Computer Science}, pages 124--134, 1994.

\bibitem{soucy2015horizontal}
S.~M. Soucy, J.~Huang, and J.~P. Gogarten.
\newblock Horizontal gene transfer: building the web of life.
\newblock {\em Nature Reviews Genetics}, 16(8):472--482, 2015.

\bibitem{van2022embedding}
L.~van Iersel, M.~Jones, and M.~Weller.
\newblock Embedding phylogenetic trees in networks of low treewidth.
\newblock {\em arXiv preprint arXiv:2207.00574}, 2022.

\bibitem{van2018unrooted}
L.~Van~Iersel, S.~Kelk, G.~Stamoulis, L.~Stougie, and O.~Boes.
\newblock On unrooted and root-uncertain variants of several well-known
  phylogenetic network problems.
\newblock {\em Algorithmica}, 80(11):2993--3022, 2018.

\bibitem{van2010locating}
L.~van Iersel, C.~Semple, and M.~Steel.
\newblock Locating a tree in a phylogenetic network.
\newblock {\em Information Processing Letters}, 110(23):1037--1043, 2010.

\bibitem{weller2018linear}
M.~Weller.
\newblock Linear-time tree containment in phylogenetic networks.
\newblock In {\em RECOMB International Conference on Comparative Genomics},
  pages 309--323. Springer, 2018.

\end{thebibliography}

\appendix
\section{Python Program that Generates the QUBO Formulation for an Instance of Tree Containment}\label{appendix}

\red{In addition to the two code listings below, the jupyter notebook at \url{https://colab.research.google.com/drive/1YvyVNXhBcnAItv-_XqGwRcoFOWPk46_T?usp=sharing} illustrates how to convert an instance of the  Tree Containment problem to a QUBO and analyze the output.}

\begin{lstlisting}[caption={Sample input of a phylogenetic tree and network. Remove the comments before passing to the program.}, upquote=true]
5       # order of T
1 4     # adjacency list of T
2 3



7       # order of N
1 2     # adjacency list of N
3 4
3 6
5



3       # number of leaves
2 8     # which leaf in T corresponds to which leaf in N.
3 3     
4 4
\end{lstlisting}

\begin{lstlisting}[language=Python, caption=Python code to generate the QUBO for an instance of Tree Containment, label={lst:listing-python}, upquote=true]
import networkx as nx
import sys
import math
from pyqubo import Array


def read_input():
    n = int(sys.stdin.readline().strip())
    T = nx.empty_graph(n, create_using=nx.DiGraph())
    for u in range(n):
        neighbors = sys.stdin.readline().split()
        for v in neighbors:
            T.add_edge(u, int(v))

    n = int(sys.stdin.readline().strip())
    N = nx.empty_graph(n, create_using=nx.DiGraph())
    for u in range(n):
        neighbors = sys.stdin.readline().split()
        for v in neighbors:
            N.add_edge(u, int(v))

    # leaf_correspondence represents which leaf in T corresponds to which leaf in N.
    # e.g. {5:7, 2:3} means leaf with index 5 in T corresponds to leaf with index 7 in N.
    leaf_correspondence = {}
    n = int(sys.stdin.readline().strip())
    for i in range(n):
        leaves = sys.stdin.readline().split()
        leaf_correspondence[int(leaves[0])] = int(leaves[1])
    return (T, N, leaf_correspondence)


def edge(u, v, graph):
    if v in list(graph.neighbors(u)):
        return True
    return False


def get_tree_vertices(graph):
    tree_vertices = []
    for i in range(graph.order()):
        if graph.out_degree(i) == 2:
            tree_vertices.append(i)

    return tree_vertices


def get_reticulation_vertices(graph):
    reticulation_vertices = []
    for i in range(graph.order()):
        if graph.in_degree(i) == 2:
            reticulation_vertices.append(i)

    return reticulation_vertices


def get_leaf_vertices(graph):
    leaf_vertices = []
    for i in range(graph.order()):
        if graph.out_degree(i) == 0:
            leaf_vertices.append(i)

    return leaf_vertices


def generate_qubo(T, N, leaf_correspondence):
    a = T.order() + 1
    if T.order() > N.order():
        raise Exception('N must have at least as many vertices as T.'
                        )
    elif T.order() == N.order():
        b = N.order()
    else:
        b = N.order() + 1 + math.floor(math.log2(N.order() - T.order()))
    x = Array.create('x', shape=(a, b), vartype='BINARY')
    z = Array.create('z', shape=(T.order(), N.order()), vartype='BINARY'
                     )
    zhat = Array.create('zhat', shape=(T.order(), 2 * N.order()),
                        vartype='BINARY')

    A = 2 * N.order()
    B = 4 * N.order() ** 2 * T.order() ** 2

    P_1 = 0
    for i in range(T.order()):
        temp = 1
        for j in range(N.order()):
            temp += -x[i, j]
        if i != 0:
            for r in range(N.order(), b):
                temp += 2 ** (r - N.order()) * x[i, r]
        P_1 += temp ** 2

    P_2 = 0
    for j in range(N.order()):
        temp = 1
        for i in range(T.order() + 1):
            temp += -x[i, j]
        P_2 += temp ** 2

    P_3 = 0
    for i in range(T.order()):
        for j in get_tree_vertices(N):
            children = list(N.neighbors(j))
            P_3 += x[i, children[0]] * x[i, children[1]] - 2 * x[i,
                    children[0]] * z[i, j] - 2 * x[i, children[1]] \
                * z[i, j] + 3 * z[i, j]

    P_4 = 0
    for i in range(T.order()):
        for j in get_tree_vertices(N):
            P_4 += x[i, j] * z[i, j]

    P_5 = 0
    for i in range(T.order()):
        for j in get_reticulation_vertices(N):
            parents = list(N.predecessors(j))
            P_5 += x[i, parents[0]] * x[i, parents[1]] - 2 * x[i,
                    parents[0]] * z[i, j] - 2 * x[i, parents[1]] * z[i,
                    j] + 3 * z[i, j]

    P_6 = 0
    for i in range(T.order()):
        for j in get_reticulation_vertices(N):
            P_6 += x[i, j] * z[i, j]

    P_7 = 0
    for i in range(T.order()):
        for l in range(T.order()):
            if i != l and edge(i, l, T):
                for j in get_tree_vertices(N):
                    P_7 += x[i, j] * z[l, j]

    P_8 = 0
    for i in range(T.order()):
        if i in get_leaf_vertices(T): continue
        for j in get_tree_vertices(N):
            children = list(N.neighbors(j))
            P_8 += x[i, j] * x[i, children[0]] - 2 * x[i, j] \
                * zhat[i, 2 * j] - 2 * x[i, children[0]] \
                * zhat[i, 2 * j] + 3 * zhat[i, 2 * j]
            P_8 += x[i, j] * x[i, children[1]] - 2 * x[i, j] \
                * zhat[i, 2 * j + 1] - 2 * x[i, children[1]] \
                * zhat[i, 2 * j + 1] + 3 * zhat[i, 2 * j + 1]

    P_9 = 0
    for i in range(T.order()):
        for l in range(T.order()):
            if i != l and edge(i, l, T):
                for j in get_tree_vertices(N):
                    children = list(N.neighbors(j))
                    P_9 += zhat[i, 2 * j] * x[l, children[1]] + zhat[i,
                            2 * j + 1] * x[l, children[0]]

    P_10 = 0
    for i in get_leaf_vertices(T):
        P_10 += 1 - x[i, leaf_correspondence[i]]

    P_11 = 0
    for i in range(T.order()):
        for j in range(N.order()):
            temp = 1
            for k in range(N.order()):
                if k != j and edge(j, k, N):
                    temp += -x[i, k]
            P_11 += x[i, j] * temp
    P_11 += -T.order()

    P_12 = 0
    for i in range(T.order()):
        for l in range(T.order()):
            if i != l and edge(i, l, T):
                temp = 1
                for j in range(N.order()):
                    for k in range(N.order()):
                        if j != k and edge(j, k, N):
                            temp += -x[i, j] * x[l, k]
                P_12 += temp

    H = B * (P_1 + P_2 + P_3 + P_4 + P_5 + P_6 + P_7 + P_8 + P_9
             + P_10) + A * P_11 + P_12

    model = H.compile()
    (qubo, offset) = model.to_qubo()
    return (qubo, model, offset)


(T, N, leaf_correspondence) = read_input()
(qubo, model, offset) = generate_qubo(T, N, leaf_correspondence)
\end{lstlisting}

\end{document}